\documentclass[a4paper,USenglish,cleveref,autoref,thm-restate]{lipics-v2021}

\pdfoutput=1 

\title{Local Computation Algorithms for Hypergraph Coloring -- following Beck's approach}

\author{Andrzej Dorobisz}{Theoretical Computer Science Department, Faculty of Mathematics and Computer Science, Jagiellonian University, Krak\'{o}w, Poland}{andrzej.dorobisz@tcs.uj.edu.pl}{https://orcid.org/0000-0002-0910-4370}{}

\author{Jakub Kozik}{Theoretical Computer Science Department, Faculty of Mathematics and Computer Science, Jagiellonian University, Krak\'{o}w, Poland}{jakub.kozik@uj.edu.pl}{https://orcid.org/0000-0002-1362-7780}{}

\authorrunning{A. Dorobisz and J. Kozik}
\Copyright{Andrzej Dorobisz and Jakub Kozik}



\category{} 

\relatedversion{} 

\funding{This work was partially supported by Polish National Science Center (2016/21/B/ST6/02165)}

\EventEditors{John Q. Open and Joan R. Access}
\EventNoEds{2}
\EventLongTitle{42nd Conference on Very Important Topics (CVIT 2016)}
\EventShortTitle{CVIT 2016}
\EventAcronym{CVIT}
\EventYear{2016}
\EventDate{December 24--27, 2016}
\EventLocation{Little Whinging, United Kingdom}
\EventLogo{}
\SeriesVolume{42}
\ArticleNo{23}

\usepackage{xcolor}
\newcommand{\gray}{\textcolor{gray}}

\newcommand{\FullShort}[2]{\ignorespaces #1}    
\hideLIPIcs 
\nolinenumbers 

\usepackage[ruled,vlined,titlenumbered,linesnumbered,noend,norelsize]{algorithm2e}
\DontPrintSemicolon
\SetNlSty{footnotesize}{}{}
\IncMargin{0.5em}
\SetNlSkip{1em}
\SetKwIF{If}{ElseIf}{Else}{if}{then}{else if}{else}{endif}
\SetKw{Return}{return}
\SetKw{Pick}{pick}
\SetKw{FAIL}{FAIL}
\SetKw{SUCCESS}{SUCCESS}
\SetNlSkip{2em}
\newenvironment{algorithm-hbox}{\hbadness=10000\begin{algorithm}}{\end{algorithm}}

\SetAlgorithmName{Listing}{Listings}

\newcommand{\codeCommentLine}[1]{\gray{\tcp{#1}}}
\newcommand{\codeCommentRight}[1]{\gray{\tcp*{#1}}}
\newcommand{\aand}{\textbf{and}}

\newcommand{\fHColoring}{\fun{hypergraph\_coloring}}
\newcommand{\fQuery}{\fun{query}}
\newcommand{\fBuildComponent}{\fun{build\_final\_component}}
\newcommand{\fExpandBadComp}{\fun{expand\_bad\_component}}
\newcommand{\fExpandViaUnsafe}{\fun{expand\_via\_unsafe}}
\newcommand{\fDetermineEdgeStatus}{\fun{determine\_edge\_status}}
\newcommand{\fExpandOrAccept}{\fun{expand\_or\_accept}}
\newcommand{\fColorComponent}{\fun{color\_final\_component}}

\SetKwFunction{FHColoring}{\fHColoring}
\SetKwFunction{FQuery}{\fQuery}
\SetKwFunction{FBuildComponent}{\fBuildComponent}
\SetKwFunction{FExpandBadComp}{\fExpandBadComp}
\SetKwFunction{FExpandViaUnsafe}{\fExpandViaUnsafe}
\SetKwFunction{FDetermineEdgeStatus}{\fDetermineEdgeStatus}
\SetKwFunction{FExpandOrAccept}{\fExpandOrAccept}
\SetKwFunction{FColorComponent}{\fColorComponent}

\newcommand{\eul}{\mathrm{e}}

\newcommand{\Ett}[2]{\mathbb{E}\left[{\textrm{#1}}\atop{\textrm{#2}}\right]}

\newcommand{\pr}[1]{\Pr{[#1]}}
\newcommand{\prt}[1]{\pr{\textrm{#1}}}

\newcommand{\ceil}[1]{\lceil{#1}\rceil}

\newcommand{\Oh}[1]{\mathcal{O}{\left(#1\right)}}

\newcommand{\polylog}{\mathrm{polylog}}

\newcommand{\event}[1]{\mathcal{X}(#1)}
\newcommand{\eventBasic}[1]{\event{#1}}

\newcommand{\fun}[1]{\textsc{#1}{}}     
\newcommand{\state}[1]{\texttt{#1}}     

\newcommand{\rc}{\mathcal{C}}           

\newcommand{\type}[1]{(#1)}
\newcommand{\typeM}{\type{M}}
\newcommand{\typeB}{\type{B}}
\newcommand{\typeBin}{\type{B_{in}}}
\newcommand{\typeBout}{\type{B_{out}}}
\newcommand{\typeU}{\type{U}}
\newcommand{\typeE}{\type{E}}

\newcommand{\ort}{\phi}
\newcommand{\lab}{\sigma}

\newcommand{\tCover}[2][]{R_{#1}(#2)}     

\newcommand{\vCompBound}{2(\Delta+1)\log(m)}  

\newcommand{\RESAMPLE}{\textsc{Resample}} 

\newcommand{\nextSection}{}           
\newcommand{\nextSectionGap}{}
\newcommand{\nextLines}{}
\newcommand{\nextLine}{}

\begin{document}

\maketitle

\begin{abstract}    We investigate local computation algorithms (LCA) 
        for two-coloring of $k$-uniform hypergraphs.
    We focus on hypergraph instances that satisfy strengthened assumption of 
        the Lov\'{a}sz Local Lemma of the form
        $2^{1-\alpha k} (\Delta+1) \eul < 1$,
        where $\Delta$ is the bound on the maximum edge degree.
    The main question which arises here is
        for how large $\alpha$ there exists 
            an LCA that is able to properly color such hypergraphs
                in polylogarithmic time per query.
    We describe briefly how 
        upgrading the classical sequential procedure of Beck from 1991
             with Moser and Tardos' \RESAMPLE{}
                yields polylogarithmic LCA that works for $\alpha$ up to $1/4$.
    Then, we present an improved procedure that
        solves wider range of instances by allowing $\alpha$ up to $1/3$.

\end{abstract}

\section{Introduction}

The problem of hypergraph coloring often serves as a benchmark for various probabilistic techniques.
The task is to answer whether there exist (or to explicitly find) \emph{a proper coloring},
    that is, such an assignment of colors to the vertices of a hypergraph
    that no edge contains vertices all of the same color.
In fact, the problem of two-coloring\footnote{
        In two-coloring problem we can assign to each vertex one of two available colors.
    }
    of linear hypergraphs was one of the main motivations for introducing Local Lemma
    in the seminal paper of Erd\H{o}s and Lov\'{a}sz \cite{EL1975}.
It is well known that determining whether the given hypergraph admits proper two-coloring
    is NP-complete \cite{Lov73}.
This result holds even for hypergraphs with all edges of size 3.
In this work, we discuss sublinear algorithms for two-coloring of uniform hypergraphs
    within the framework of Local Computation Algorithms.

We are going to work with $k$-uniform hypergraphs\footnote{
    In \emph{$k$-uniform hypergraph} each edge contains exactly $k$ vertices.
}.
For the rest of the paper, 
    $n$ is used to denote the number of vertices of considered uniform hypergraph, 
    $m$ its number of edges, 
    and $k$ size of the edges.
We assume that $k$ is fixed (but sufficiently large to avoid technical details) 
    and that $n$ tends to infinity.
For a fixed hypergraph, we denote by $\Delta$ its maximum edge degree.
In the instances with which we are going to work, 
    $\Delta$ is bounded by a function of $k$, so in terms of $n$ it is $\Oh{1}$.
This implies that the number of edges $m$ is at most linear in $n$.
We also  assume that the considered hypergraphs do not have isolated vertices.
Then, we also have $m = \Theta(n)$.

\subsection{Local Computation Algorithms}
Rubinfeld, Tamir, Vardi and Xie proposed in \cite{fast-local} 
    a general model of sublinear sequential algorithms
        called Local Computation Algorithms (LCA).
The model is intended to capture the situation 
    where some computation has to be performed on a large instance
    but, at any specific time, only parts of the answer are required.
The interaction with a local computation algorithm is organized in the sequence of queries 
    about fragments of a global solution.
The algorithm shall answer each consecutive query in sublinear time 
    (wrt the size of the instance),
    systematically producing a partial answer 
        that is consistent with some global solution.
The model allows for randomness, and algorithm may occasionally fail.

For example, for the hypergraph two-coloring problem, the aim of an LCA procedure is
    to find a proper coloring of a given hypergraph.
The algorithm can be queried about any vertex, and in response, 
    it has to assign to the queried vertex one of the two available colors.
For any sequence of queries, with high probability, 
    it should be possible to extend the returned partial coloring to a proper one.

Formally, for a fixed problem, 
    a procedure is a \emph{$(t,s,\delta)$-local computation algorithm},
    if for any instance of size $n$ and any sequence of queries,
        it can consistently answer each of them in time $t(n)$
        using up to $s(n)$ space for computation memory.
The time $t(n)$ has to be sublinear in $n$,
    but a polylogarithmic dependence is desirable. 
The value $\delta(n)$ shall bound the probability of failure 
    for the whole sequence of queries.
It is usually demanded to be small.
The computation memory, the input, and the source of random bits 
    are all represented as tapes with random access 
    (the last two are not counted in $s(n)$ limit).
The computation memory can be preserved between queries.
In particular, it can store some partial answers determined in the previous calls.
For the precise general definition of the model consult \cite{fast-local}.

A procedure is called \emph{query oblivious} if the returned solution does not depend
    on the order of the queries  (i.e. it depends only on the input and the random bits).
It usually indicates that the algorithm uses computation memory 
    only to answer the current query 
    and that there is no need to preserve information between queries.
It is a desirable property, since it allows to run queries to algorithm in parallel.
In a follow-up paper \cite{space-eff}, Alon, Rubinfeld, Vardi, and Xie presented 
    generic methods of 
        removing query order dependence
        and reducing necessary number of random bits
    in LCA procedures.
In the same paper, these techniques were applied to the example procedures
    (including hypergraph coloring)
    from \cite{fast-local} converting them to query oblivious LCAs.
The improved procedures work not only in polylogarithmic time but also in polylogarithmic space.
Mansour, Rubinstein, Vardi, and Xie in \cite{MRVX12} improved analysis of this approach.

\subsection{Constructive Local Lemma and LCA}

The Lov\'{a}sz Local Lemma (LLL) is one of the most important tools in the field of local algorithms.
In its basic form, it allows one to non-constructively prove the existence of
    combinatorial objects 
        omitting a collection of undesirable properties, so-called bad events.
A~brief introduction to this topic and a summary of various versions of LLL
    can be found in the recent survey by Farag\'{o} \cite{survey-2021}.

For a fixed $k$-uniform hypergraph, let $p=2^{-k}$ denote the probability that, 
    in a uniformly random coloring, a~fixed edge is monochromatic in a~specific color.
A straighforward application of the symmetric version of Local Lemma (see e.g., \cite{survey-2021})
proves that the condition 
$
    2 p \; (\Delta+1) \; \eul < 1,
$
is sufficient for a hypergraph with the maximum edge degree $\Delta$, 
    to be two-colorable.

For many years, Local Lemma resisted attempts to make it efficiently algorithmic.
The first breakthrough came in 1991, when Beck \cite{Beck91},
    working on the example of hypergraph two-coloring,
    showed a method of converting some of LLL existence proofs 
            into polynomial-time algorithmic procedures.
However, in order to achieve that, 
    the assumptions of Local Lemma had to be strengthened and took form
\begin{equation} \label{eq:parametrized-LLL-condition}
    2 \; p^{\alpha} \; (\Delta+1) \; \eul < 1.
\end{equation}
For $\alpha=1$ the inequality reduces to the standard assumption.
The above inequality constraints $\Delta$, and the constraint becomes more restrictive as $\alpha$ gets smaller.
The original proof of Beck  worked for $\alpha < 1/48$.
From that time, a lot of effort has been put into studying applications to specific problems and pushing $\alpha$ forward, 
    as close as possible to standard LLL criterion \cite{alon,MR98,non-uniform,Sri08,Mos08}.

The next breakthrough was made by Moser in 2009.
In cooperation with Tardos, Moser's ideas have been recasted in \cite{moser-tardos} into general constructive formulation of the lemma.
They showed that,
    assuming so called variable setting of LLL,
    a natural randomized procedure called \RESAMPLE\footnote{
        As long as some bad events are violated, 
         the procedure picks any such event and resamples all variables on which that event depends.
        }
    quickly finds an evaluation of involved random variables for which none of the bad events hold.
They also proved that, in typical cases,
    the expected running time of the procedure is linear in the size of the instance.
For the problem of two-coloring of $k$-uniform hypergraphs,
    the total expected number of resamplings is bounded by $m/\Delta$
    (see Theorem 7 in \cite{survey-2021}).

Adjusting constructive LLL to LCA model remains one of the most challenging problems in the area.
It turns out, however, that previous results on algorithmization of Local Lemma
        can be adapted in the natural way.
In fact, the first LCA algorithm for the hypergraph coloring from \cite{fast-local},
    is built on the variant of Beck's algorithm 
        that is described in the book by Alon and Spencer \cite{ProbMeth-v2}.
That version works for $\alpha < 1/11$, and runs in polylogarithmic time per query.
Later refinements focused on optimizing space and time requirements
    (\cite{space-eff}, \cite{MRVX12}),
    however, for polylogarithmic LCAs the bound on $\alpha$ has not been improved.
In a recent work, Achlioptas, Gouleakis, and Iliopoulos \cite{AGI} showed
    how to adjust \RESAMPLE{} to LCA model.
They did not manage, however, to obtain a polylogarithmic time.
Their version answers queries in time $t(n)=n^{\beta(\alpha)}$.
They establish some trade-off between the bound on $\alpha$ and the time needed to answer a query.
In particular, when $\alpha$ approaches $1/2$ then $\beta(\alpha)$ tends to $1$, which results in a very weak bound on the running time per query.

\subsection{Main result}  \label{ssec:main_result}

Our research focuses on the following general question 
    in the area of local constructive versions of the Lov\'{a}sz Local Lemma:
    up to what value of $\alpha$ there exists a polylogarithmic LCA
        for the problem of two-coloring of $k$-uniform hypergraphs satisfying condition
            $2(\Delta+1) \eul < 2^{\alpha{k}}$.
We prove the following theorem:

\begin{theorem}[main result] \label{thm:main}
    For every $\alpha < 1/3$ and all large enough $k$,
    there exists a local computation algorithm that,
    in polylogarithmic time per query,
        with probability $1-O(1/n)$ solves the problem of two-coloring
        for $k$-uniform hypergraphs with maximum edge degree $\Delta$,
            that satisfies
   $2 \eul (\Delta+1) < 2^{\alpha k}.$
\end{theorem}
Within the notation of \cite{fast-local} we present
    $(\polylog(n), \Oh{n}, \Oh{1/n})$-local computation algorithm
    that properly colors hypergraphs that satisfy the above assumption.
Our algorithm is not query oblivious.
Moreover, typical methods of eliminating the dependence on the order of queried vertices
    do not seem to be applicable without sacrificing constant $\alpha$.
\FullShort{
    More technical and precise statement of our main result is presented
    in Appendix \ref{sec:proof} as Theorem \ref{thm:main-technical}.
}{
    Consult the full version of this paper for the complete proof of the theorem.
}

For comparison,
    Alon et al. \cite{space-eff} after Rubinfeld et al. \cite{fast-local}
    present a query oblivious
    $(\polylog(n), \polylog(n), \Oh{1/n})$-local computation algorithm
    working for hypergraphs satisfying
\begin{align} \label{eq:fast-local-conditions}
\nonumber    16\;\Delta(\Delta-1)^{3}(\Delta+1) &< 2^{k_{1}},\\ 
             16\;\Delta(\Delta-1)^{3}(\Delta+1) &< 2^{k_{2}},\\
\nonumber    2e(\Delta+1) &< 2^{k_{3}},
\end{align}
where $k_{1}, k_{2}$ and $k_{3}$ are positive integers such that 
$k = k_{1}+k_{2}+{k_3}$.
These assumptions correspond to $\alpha < 1/11$.

The analysis of the LCA procedure from \cite{space-eff}
    guarantees only that the running time is 
    of the order $\Oh{\log^{\Delta}(n)}$.
Mansour et al. in \cite{MRVX12}
    focus on improving time and space bounds within polylogarithmic class,
        removing the dependency on the maximal edge degree from the exponent.
They obtain an LCA working in $\Oh{\log^{4}(n)}$ time and space,
    assuming that $k \geq 16 \log(\Delta) + 19$,
        so it requires even stronger bound on $\alpha$.

\subsection{LOCAL distributed algorithms}

The model of Local Computation Algorithms is related to the classical model of local distributed computations by Linial \cite{Lin87} (called LOCAL).
For comparison of these two models, see work of Even, Medina, and Ron \cite{EMR18}.
Chang and Pettie observed recently in \cite{CP17} that within LOCAL model,
    the general problem of solving Local Lemma instances with a~dependency graph of bounded degree
    is in some sense complete for a large class of problems
    (these are the problems which can be solved in sublogarithmic number of rounds).
They also conjectured that for sufficiently strengthened condition of Local Lemma
    (like taking small enough $\alpha$ in (\ref{eq:parametrized-LLL-condition}))
    there exists a distributed LOCAL algorithm that solves the problem in $\Oh{\log\log n}$ rounds.
The straightforward simulation of such an algorithm within LCA framework would yield a procedure that,
   at least for fixed maximum degree, answers queries in polylogarithmic time.

Recently, progress towards this conjecture has been made by Fischer and  Ghaffari \cite{FG17}, who proved that
    there exists an algorithm for Local Lemma instances that works in $2^{\Oh{\sqrt{\log\log n}}}$ rounds.
The influence of the degree of underlying dependency graph on running time has been later improved by Ghaffari, Harris and Kuhn in \cite{GHK18}.
In particular, for sufficiently constrained problem of hypergraph two-coloring,
    that result allows one to obtain an LCA procedure that answers queries in sublinear time.
The time, however, would be superpolylogarithmic.
Moreover, the necessary strengthening of Local Lemma assumptions 
    appears to be much stronger than the one required to apply the result of Rubinfeld et al. \cite{fast-local}.

The possibility of simulation of LOCAL algorithms within LCA model implies that if Chang and Pettie conjecture holds,
    then any problem satisfying sufficiently strengthened LLL conditions can be solved in LCA model in polylogarithmic time per query.
We can therefore formulate a weaker conjecture that for some $\alpha$ every such $\alpha$-strengthened problem 
    can be solved in LCA in polylogarithmic time per query.
For the specific problem of hypergraph coloring,  this property is known to hold.
We can, however, ask what is the maximum such $\alpha$ for a fixed problem.
That is precisely the general problem stated at the beginning of Section \ref{ssec:main_result}. 
It is interesting to note that our algorithms make essential use of the sequential nature of LCA.
For that reason, they cannot be translated to $\Oh{\log\log n}$ LOCAL algorithms.
This also illustrates an important difference between the models.

\section{Main techniques and ideas of the proof} \label{sec:ideas}

The algorithmic procedure of Beck \cite{Beck91} is divided into two phases.
In the first one, which we call \emph{the shattering phase}, it builds a random partial coloring
    that guarantees that a fraction of all edges are already properly colored.
Moreover, the edges which are not yet taken care of have sufficiently many 
    non-colored vertices to make sure that the partial coloring can be completed to a proper one.
They also form connected components of logarithmic sizes
    which can be colored independently.
Then, in the second phase, which we call \emph{the final coloring phase},
    an exhaustive search is used to complete the coloring of each component.
This results in a sequential procedure with polynomial running time.
In order to reduce the running time to almost linear,
    the shattering phase can be applied twice.
Then, the final components w.h.p. are of size $\Oh{\log\log(n)}$.
The polylogarithmic LCA procedure for hypergraph coloring from \cite{fast-local}
    followed that approach and simulates locally two shattering phases and an exhaustive search
        when answering a single query.
Division into these three phases is directly reflected 
    in the conditions (\ref{eq:fast-local-conditions}) required by the procedure.

While it is not known whether it is possible
    to design an LCA algorithm based solely on \RESAMPLE, 
    combining it with previous local algorithms brings significant improvements.
It turns out that, within polylogarithmic time, after only one shattering phase, the coloring 
    can be completed with the use of \RESAMPLE.
This simple modification, with slightly improved analysis,
    is sufficient to derive Theorem \ref{thm:main} for $\alpha \leq 1/4$.
        This is our first contribution.
That procedure provides  
    a reference point for explaining the intuitions and motivations that underlie 
    the further improvements that we derive.
In particular, we define a notion of \emph{component-hypergraph}
    that allows for a more fine-grained analysis of the components
         of the residual hypergraph.
For that reason, we present our base algorithm in detail in Section \ref{sec:base}.

The first modification that we make in order to improve the base algorithm is that within the shattering phase
    we sample colors for all vertices.
Then, for some vertices, the color is final,
    and for others, it is allowed to change the assigned color in the final coloring phase.
Coloring all the vertices during the first phase
    somehow blurs the border between the shattering and final coloring phases.
Its main purpose is to enable a more refined partition of the residual hypergraph
    into independent fragments.
It also allows to determine  some components of the residual hypergraph
    for which no recoloring would be necessary.
This corresponds to a situation in which the first sampled colors 
        in \RESAMPLE{} happen to define a proper coloring.
Altogether, we managed to significantly reduce the pessimistic size of the independent fragments colored in the final coloring phase,
    which enables further relaxation of the necessary conditions on $\alpha$ to $\alpha<1/3$.
The improved procedure is described in Section \ref{sec:main}.

In order to analyze the procedures,
    we employ a common technique of  
        associating some tree-like \emph{witness structures}
    with components that require recoloring.
Every such structure describes a collection of events associated with some edges of the hypergraph.
All these events are determined by the colors assigned in the shattering phase.
For the base algorithm, these structures are quite typical.
However, in order to achieve the better bound on $\alpha$,
    we developed more sophisticated structures that are capable of tracking different
        kinds of events, which can also depend on the colors that are allowed to be recolored.
Different kinds of events come with different bounds on probability.
An important aspect of the analysis concerns
    amortization of different kinds of events within a single structure.
The construction of these structures is our main technical contribution.
\FullShort{
    We described it in detail in Appendix \ref{sec:proof-witness}.
}{
    Its detailed description can be found in the full version of the paper.
}

We finally note that, while our methods are not general enough 
    to work for all instances satisfying the strengthened assumptions of LLL,
    they can be applied to a number of 
    problems similar to hypergraph coloring, like, e.g. $k$-SAT.

\section{Establishing base result} \label{sec:base}

In this section we show
    how the Beck's algorithm can be combined with \RESAMPLE{} to construct
        a local computation algorithm that works in polylogarithmic time per query
            for $\alpha$ up to $1/4$.
In other words, 
    we prove Theorem \ref{thm:main}
        under the stronger assumption that $\alpha \leq 1/4$.
To keep the exposition simple,
    we first present a global randomized algorithm.
Then, we comment on how to adapt this procedure to LCA model.
\FullShort{
    The analysis of the procedure can be found in Appendix \ref{sec:base-analysis}.
}{
    The complete analysis of the procedure can be found in the full version of this paper.
}

Let $H=(V,E)$ be a hypergraph that satisfies the assumptions of Theorem \ref{thm:main}
    for a fixed $\alpha \leq 1/4$.
For technical convenience, we assume that $\alpha k$ is an integer\footnote{
        In fact, for the given $k$ it is only reasonable to take $\alpha$ in the form of $t/k$,
            where $t$ is an integer $2 \leq t \leq k$.
    }.
By assigning a random color, we mean choosing uniformly one of the two available colors.
For a set of edges $S$, by $V(S)$ we mean all vertices covered by the edges from $S$.
For an edge $f$, $N(f)$ denotes the set of edges intersecting $f$.
We use a naming convention that is similar to other works on the subject
    -- in particular, 
    our view of Beck's algorithm is influenced by 
        its descriptions by
            Alon and Spencer \cite{ProbMeth-v2}
            and Molloy and Reed \cite{GraphCol},
        as well as LCA realization given in \cite{fast-local}.

\subsection{Global coloring procedure} \label{sec:base-alg}

The algorithm starts with choosing an arbitrary order of vertices.
Then, it proceeds in two phases: \emph{the shattering phase} and \emph{the final coloring phase}.
The shattering phase colors some vertices of the input hypergraph
    and then splits the edges of the hypergraph that are not properly colored yet into
        \emph{final components} -- subhypergraphs that can be colored independently.
The final coloring phase completes the coloring by
    considering the final components separately, one by one.

\subsubsection{The shattering phase}

The procedure processes vertices sequentially according to the fixed ordering.
For every
    vertex, it either
        assigns a random color to the vertex
    or
        leave it non-colored in case it belongs to a \emph{bad} edge.
An edge is called \emph{bad} if it contains $(1-\alpha)k$ colored vertices
    and is still not colored properly (that is,  all these vertices have the same color).
Once an edge becomes bad, no more vertices from that edge will be colored
    -- such vertices are called \emph{troubled}.
Vertices with assigned colors are called \emph{accepted}.

Upon completion of the shattering phase, there are three types of edges:
\begin{itemize}
    \item \emph{safe edges} -- properly colored by the accepted vertices,
    \item \emph{bad edges} -- containing exactly $(1-\alpha)k$ accepted vertices,
                              all of the same color,
    \item \emph{unsafe edges} -- containing fewer than $(1-\alpha)k$ accepted vertices,
                                 all of the same color.
\end{itemize}
Observe that in the resulting (partial) coloring,
    every edge that is not colored properly          
        has at least $\alpha{k}$ troubled vertices,
            which will be colored in the next phase.
Note also that it might happen that some unsafe edge has no colored vertices at all.

The colors of accepted vertices are not going to be changed, so the safe edges are already taken care of.
Therefore, we focus on bad and unsafe edges.
Let $E_{bad}$ denote the set of all bad edges.
Consider hypergraph $(V(E_{bad}), E_{bad})$.
It is naturally decomposed into connected components.
\begin{definition} \label{def:bad-comp}
    Every component of the hypergraph $(V(E_{bad}), E_{bad})$ is called a \emph{bad-component}.
\end{definition}
Note that every troubled vertex belongs to some bad-component.
On top of them we build an abstract structure
    to express dependencies between bad-components through unsafe edges.

\begin{definition} \label{def:comp-hg}
    A \emph{component-hypergraph} is constructed as follows:
        its vertices are bad-components of $H$
        and for every unsafe edge $f$ intersecting more than one bad-component,
            an edge that contains all bad-components intersected by $f$ is added to it.
\end{definition}

For each connected component of the component-hypergraph
    (that is, a maximal set of bad-components that is connected in the component-hypergraph)
    we construct a \emph{final component} by taking the union of those bad-components
(hence a final component is a subhypergraph of $H$).
The shattering phase is \emph{successful} if
    each final component contains at most $\vCompBound$ bad edges.
If this is not the case, the procedure declares a failure. 
It turns out that this is very unlikely to happen.

\subsubsection{The final coloring phase}

For each final component $\rc$ determined during the shattering phase,
    we add to $\rc$ all unsafe edges intersecting it,
    and then, we restrict $\rc$ to troubled vertices\footnote{
        Restriction of $H=(V,E)$ to $V' \subseteq V$ is defined as 
            $H' = (V', \{e \cap V' | ~e \in E, e \cap V' \neq \emptyset \})$.
    }.
We obtain a hypergraph $\rc'$ containing
    at most $2(\Delta+1)^2\log(m)$ edges,
    and each of them has at least $\alpha{k}$ vertices.
The maximum edge degree in $\rc'$ cannot be larger than $\Delta$,
    which is the maximum edge degree in $H$.
Since $2e(\Delta+1) < 2^{\alpha k}$ (by the assumptions of Theorem~\ref{thm:main}),
    Lov\'{a}sz Local Lemma ensures that $\rc'$ is two-colorable.
Hence, by the theorem of Moser and Tardos
        \RESAMPLE{} finds a~proper coloring of it
        using on average $|E(\rc')|/\Delta$ resamplings
        (see Theorem 7 in \cite{survey-2021}).

When the final coloring phase is over, all final components are properly colored.
Since each bad or unsafe edge is dealt within some final component,
    and each safe edge was properly colored during the shattering phase,
    it is now guaranteed that the constructed coloring is proper for the whole $H$.

\subsection{LCA realization} \label{sec:base-lca}

We employ quite standard techniques to obtain an LCA realization of the described algorithm.
We articulate it below to provide a context for the description of our main algorithm. 
An important property of the described procedure is that
    the ordering of vertices does not have to be fixed a priori.
In fact it can be even chosen in an on-line manner by an adversary.
Following \cite{fast-local},
    we are going to exploit the freedom of choice of ordering.
The LCA version of the algorithm is going to simulate the global version run with a specific ordering.
That ordering is constructed dynamically during the evaluation and is driven by the queries.
Apart from some minor adjustment (resulting from adaptation to LCA model)
    when the algorithm is queried about vertex $v$,
    it performs all the work of the standard algorithm
        needed to assign a final color to $v$.
The LCA version is presented in Listings
        \ref{alg:base-lca-query},
        \ref{alg:base-lca-build-component},
        \ref{alg:base-lca-subprocedures},
    and \ref{alg:base-lca-color-component}.
All colors assigned during work of the algorithm are stored
    in the computation memory (which is preserved between queries).
For convenience, we also store there the status of each vertex 
    -- \emph{uncolored}, \emph{accepted} or \emph{troubled}.
Initially all vertices are uncolored.

\begin{algorithm}
    \caption{LCA for uniform hypergraph coloring -- main function} \label{alg:base-lca-query}
    \DontPrintSemicolon
    \SetKwProg{Fn}{Procedure}{:}{}
    \Fn{\FQuery{$v$ - vertex}}{
        \If {$v$ is \state{uncolored}}  {
            \If{all edges containing $v$ are not bad}{
                assign a random color to $v$ \aand{} 
                    mark it as \state{accepted}        \codeCommentRight{shattering}
            } \lElse {
                mark $v$ as \state{troubled}
            }
        }
        \If {$v$ is \state{troubled}}  {
                $\rc_v \gets $ \FBuildComponent{$v$}   \codeCommentRight{shattering}
                \FColorComponent{$\rc_v$}              \codeCommentRight{final coloring} 
        }
        \Return color assigned to $v$
  }
\end{algorithm}

\subsubsection{\fQuery}
When a vertex $v$ has been already marked as accepted, its color is immediately returned.
If it has not been processed before, the algorithm checks whether $v$ belongs to any bad edge
    (that requires inspecting the current statuses of all the edges that contain $v$).
If not,
    a random color is assigned to $v$, the vertex is marked as accepted, and the procedure returns the assigned color.
On the other hand, when $v$ belongs to a bad edge,
    it is marked as troubled.
The algorithm then determines the final component containing $v$
    in procedure \FBuildComponent.
These steps can be viewed as the shattering phase.
Afterwards, the final coloring phase is performed for the final component
    in procedure \FColorComponent.

\begin{algorithm}
    \caption{Building the final component for $v$ that belongs to some bad edge} \label{alg:base-lca-build-component}
    \DontPrintSemicolon
    \SetKwProg{Fn}{Procedure}{:}{}
    \Fn{\FBuildComponent{$v$ - troubled vertex}}{
        $B \gets \emptyset$
            \codeCommentRight{initialize set of bad edges of the component}
        $U \gets \emptyset$
            \codeCommentRight{initialize set of unsafe edges to process}
        $e \gets$ any bad edge containing $v$ \;
        mark $e$ as explored \aand{} run \FExpandBadComp{$e$, $B$, $U$} \;
        \codeCommentLine{process surrounding unsafe edges}
        \While{$U$ is not empty}{
                $f \gets$ next edge from $U$ (remove it from $U$)\;
                \FExpandViaUnsafe{$f$, $B$, $U$}
        }
        \codeCommentLine{return hypergraph built on set of bad edges}
        \Return $\rc = (V(B), B)$
    }
\end{algorithm}

\subsubsection{\fBuildComponent}
This procedure builds the set $B$ of bad edges of the final component of $v$,
    exploring the line graph of $H$\footnote{
        The line graph $L(H)$ is the graph built on $E(H)$
        in which two distinct vertices (representing edges of $H$)
        are adjacent if the corresponding edges intersect.
}.
It uses a temporary flag \emph{explored} to mark visited edges
    (this flag is not preserved between queries).
The construction starts from a bad edge containing troubled vertex $v$
    and expands it to a bad-component.
Then, as long as possible,
    set $B$ is extended by edges of neighboring bad-components,
        which can be reached through unsafe edges adjacent to $B$.
If at some point the number of bad edges in $B$
    exceeds the prescribed bound $\vCompBound$,
    then the procedure declares a failure
        (note that it cannot be restarted
         since LCA model does not allow to change colors returned for previous queries).
Construction of the final component is done when
    there are no more bad edges to add.
Then, the hypergraph $\rc = (V(B), B)$ built on the collected bad edges is returned.

The expansion of bad-components is done within subprocedure \FExpandBadComp.
It starts from the given bad edge and
    explores the line graph by inspecting the adjacent edges.
For each adjacent edge, its type (safe, unsafe, or bad) is determined using \FDetermineEdgeStatus.
Determining status of an edge may require processing some uncolored vertices of that edge.
For each of them, the procedure check whether it is troubled.
If it is not, a random color is assigned to the vertex and the vertex is marked as accepted.

\begin{algorithm}
    \caption{Subprocedures for the final component construction} \label{alg:base-lca-subprocedures}
    \DontPrintSemicolon
    \SetKwProg{Fn}{Procedure}{:}{}
    \Fn{\FExpandBadComp{$e$ - bad edge, $B$ - bad edges, $U$ - unsafe edges}}{
        $Q \gets \{ e \} $
            \codeCommentRight{initialize set of bad edges to process}
        \While{$Q$ is not empty}{
            $f \gets$ next edge from $Q$ (remove it from $Q$)\;
            add $f$ to $B$ \aand{}
                \lIf{$|B|> \vCompBound$}{ \FAIL } 
            \For{$g \in N(f)$ which are not explored}{
                mark $g$ as explored       \aand{}
                \FDetermineEdgeStatus{$g$} \;
                \lIf {$g$ is \state{bad}}    { add $g$ to $Q$ }
                \lIf {$g$ is \state{unsafe}} { add $g$ to $U$ }
            }
        }
    }
    \;
    \Fn{\FExpandViaUnsafe{$f$ - unsafe edge, $B$ - bad edges, $U$ - unsafe edges}}{
                \For{$g \in N(f)$ which are not explored}{
                    \FDetermineEdgeStatus{$g$} \;
                    \If{$g$ is \state{bad}}{
                        mark $g$ as explored \aand{} run \FExpandBadComp{$g$, $B$, $U$}
                    }
                }
    }
    \;
    \Fn{\FDetermineEdgeStatus{$g$ - edge}}{
        \For{each $w$ in $g$ that is \state{uncolored} unless $g$ becomes \state{safe}}{
            \lIf{some edge containing $w$ (including $g$) is bad}{ mark $w$ as \state{troubled} }
            \lElse{ assign a random color to $w$ \aand{} mark it as \state{accepted} }
        }
        count accepted vertices and check their colors to determine status of $g$
    }
\end{algorithm}

During the expansion through unsafe edges we keep a set $U$ of not processed unsafe edges
    that intersects any edge of $B$.
As long as $U$ is not empty,
    we pick any unsafe $f$ from $U$ and process it by \FExpandViaUnsafe.
Here we determine the statuses of all edges adjacent to $f$
    and if we encounter a bad edge which is not in $B$,
        then we add it and expand a bad-component containing it.
For technical convenience, during bad-component expansion we collect
    non-explored adjacent unsafe edges and add them to $U$.

\begin{algorithm}
  \caption{Finding coloring inside the final component} \label{alg:base-lca-color-component}
  \DontPrintSemicolon
  \SetKwProg{Fn}{Procedure}{:}{}
    \Fn{\FColorComponent{$\rc$ - hypergraph}}{
        add to $\rc$ all \state{unsafe} edges intersecting $\rc$ \;
        $\rc' \gets$ restriction of $\rc$ to troubled vertices \;
        $t_{e} \gets |E(\rc')|/\Delta$  \codeCommentRight{expected time of one RESAMPLE trial}
        \For{$trial=1$ to $2 \log(m)$}{
            \codeCommentLine{RESAMPLE with limited number of steps}
            assign random colors to $V(\rc')$ \;
            \For{$step=1$ to $2 t_e$}{
                \eIf{there is monochromatic $f \in E(\rc')$}{
                    assign new random colors to all vertices of $f$
                }
                {
                    \codeCommentLine{$\rc'$ is properly colored}
                    mark all vertices of $\rc'$ as \state{accepted}
                    \aand{} \Return
                }
            }
        }
        \FAIL
    }
\end{algorithm}

\subsubsection{\fColorComponent}
Final component $\rc$ is extended with unsafe edges that intersect it.
Then it is restricted to the set of its troubled vertices.
The resulting hypergraph is denoted by $\rc'$.
The 
    algorithm tries to find a proper coloring of $\rc'$ using \RESAMPLE{} procedure.
To ensure polylogarithmic time,
    it is run only for the limited number of resampling steps.
To decrease the probability of a failure,
    the procedure may be restarted a few times.
When a proper coloring is found,
    each vertex of $\rc'$
        is marked as accepted.
From now on, all edges of $\rc$ are treated as safe.
However, if all trials were unsuccessful, the procedure declares a failure.

\nextSection
\section{Main result - algorithm} \label{sec:main}

We show how to improve the base procedure described in the previous section 
    to obtain an algorithm 
        that can be used to prove Theorem \ref{thm:main},
that is, an algorithm
    that works in polylogarithmic time per query
    on input hypergraphs that satisfy
        strengthened LLL condition~(\ref{eq:parametrized-LLL-condition}) for $\alpha < 1/3$.
Actually, our procedure can be used to find a proper coloring
    also for instances that satisfy that condition with any $\alpha \in (0,1)$,
    but the running time is not guaranteed for $\alpha\geq 1/3$.
We start with introducing the main ideas behind algorithm improvement
    and describe its global version.
Then, we discuss how to adapt it to 
    the model of the local computation algorithms,
    and finally we present a description of the LCA procedure.
\FullShort{
    The analysis of the algorithm is placed in Appendix \ref{sec:proof}.
}{
    The analysis of the algorithm can be found in the full version of this paper.
}

\subsection{A general idea}

It is a common approach in randomized coloring algorithms to start  
    from an initial random coloring and then make some correction to convert it to a proper one
        (like in \RESAMPLE{} \cite{moser-tardos} or in Alon's parallel algorithm \cite{alon}).
This is not the case of Beck's procedure, in which a proper coloring is constructed incrementally, 
    but coloring of some vertices (those marked as troubled) is postponed to the later phase.
Our approach lies somewhere in between.
We generally try to follow the latter one, but 
    we sample colors for the troubled vertices already in the shattering phase.
Such colors are considered as \emph{proposed},
    and we reserve the possibility of changing them in the final coloring phase. 
We use the information about the proposed colors
    to shrink the area that will be processed in the final coloring phase.
In particular, if we look at the colors proposed for troubled vertices,
    then only those final components that contain a monochromatic edge require recoloring.
Moreover, if we carefully track dependencies between bad-components
    (see Definition \ref{def:bad-comp}),
    it is also possible to decrease the sizes of the final components.
We explain this idea in more detail in the following subsections.

\subsubsection{Activation of bad-components}

Imagine that all the vertices were colored in the shattering phase
    and we want to determine the final components.
We look at the component-hypergraph (see Definition \ref{def:comp-hg})
    and have to decide which of the bad-components should be recolored.
We start from bad-components that are intersected by monochromatic edges
    - we mark them as \emph{initially active} and treat them as seeds of final components.
The remaining ones are currently \emph{inactive}.
Our intention is to recolor only active components in the final coloring phase.
Note that it might not be sufficient to alter the coloring in a way that makes initially active components properly colored,
    because after their recoloring, it is possible 
    that some unsafe edge which get both colors in the shattering phase becomes monochromatic.
That is why the activation has to be propagated. We use the following rule
\begin{itemize} \item
Let $A_t$ be the set of troubled vertices that are covered by active bad-components,
    and $f$ be an unsafe edge that intersects $A_t$.
If $f \setminus A_t$ is monochromatic, then 
    all inactive bad-components that intersect $f$ become active
    and all bad-components that intersect $f$ are merged into one (eventually final) component.
\end{itemize}
The above propagation rule is applied as long as possible.
When it stops, it is guaranteed that
    all monochromatic edges are inside active components 
    and all unsafe and bad edges outside of active components
        are properly colored by the vertices that are outside of active bad-components.
In particular, we can accept all the colors proposed for inactive vertices.

\subsubsection{Edge trimming}
We employ an additional technique, which can further reduce the area of the final components.
Observe that, in order to guarantee two-colorability of the final components,
    it is enough to ensure that each edge has at least $\alpha k$ vertices to recolor
        inside one final component.
It means that if some active component already contains $\alpha k$ troubled vertices 
    of some edge, then it is not necessary to propagate activation through that edge.
Thus, we can improve the propagation rule in the following way.
Consider an unsafe edge $f$ for which $f \setminus A_t$ is monochromatic
    (recall that $A_t$ denotes the set of currently active troubled vertices).
If some active component contains at least $\alpha k$ troubled vertices of $f$, 
    then $f$ is trimmed to that active component.
Otherwise, all bad-components intersected by $f$ are activated and merged into one component (as described in the previous section).

We point out that the direct inspiration for this technique came from
    the work of Czumaj and Scheideler \cite{non-uniform}
        in which the edge trimming is actively used
            during the construction of the area to be recolored.
One of the consequences of using it
    is that the shapes of the final components 
    depend on the specific order in which activation is propagated.

\subsection{Global coloring procedure}

Similarly to the base algorithm from Section \ref{sec:base-alg},
    the improved procedure performs the shattering phase and then the final coloring phase.
The former is modified according to the ideas described in the previous subsection.
In particular, we use the notions of \emph{proposed} and \emph{accepted} colors.
Pseudocode of the whole procedure can be found in Listing \ref{alg:main-global}
in Appendix \ref{sec:listings_improvedLCA}.

\subsubsection{The shattering phase}

The first part of the shattering phase is almost the same,
    except that now each vertex is colored.
The procedure processes the vertices in a fixed order,
    and for each vertex it marks it as \emph{accepted} or \emph{troubled}
        and chooses a random color.
A~vertex $v$ is accepted if, at the time of processing, 
    $v$ does not belong to any of the bad edges. Otherwise, it is troubled.
An edge becomes \emph{bad} when its set of accepted vertices reaches size $(1-\alpha)k$
    and is still monochromatic.
After processing all the vertices, 
    \emph{safe} and \emph{unsafe} edges 
        are determined in the same way as in the base algorithm.
Additionally, by a \emph{monochromatic} edge, we mean an edge for which 
    all its vertices (accepted and troubled) have the same color.
The colors of the accepted vertices are called \emph{accepted colors}.
The colors of the troubled vertices are called \emph{proposed colors}.
By accepting a color assigned to a vertex, we mean changing its status to accepted.

The next step involves determining the final components.
We work with the component-hypergraph. 
We are going to mark some bad-components and unsafe edges as \emph{active}.
By an \emph{active component}, we mean a maximal set of active bad-components
    which is connected in the component-hypergraph via active unsafe edges.
We start with marking as active
    all monochromatic unsafe edges and
    all bad-components that are intersected by any (bad or unsafe) monochromatic edge.
Let $A_t$ denote the set of troubled vertices that are currently covered by active bad-components.
Then, as long as there exists an inactive unsafe edge $f$ satisfying the following conditions:
\begin{itemize}
    \item $f$ is monochromatic outside the active troubled area
              (i.e., $f \setminus A_t$ is monochromatic), and
    \item each active component contains less than $\alpha k$ troubled vertices of $f$,
\end{itemize}
we activate $f$ and activate all bad-components intersected by $f$.
When this propagation rule can no longer be applied,
    we accept the colors of all the troubled vertices from inactive bad-components.
At that time, each active component determines a final component as the union of its bad-components.
Just like in the base algorithm, the shattering phase is \emph{successful} if
    each final component contains at most $\vCompBound$ bad edges.
Otherwise, the procedure declares a failure. 

\subsubsection{The final coloring phase}

We implement one modification at the beginning of the final coloring phase.
For each final component $\rc$, 
    we add to $\rc$ not all unsafe edges intersecting it, but
    only those that have at least $\alpha k$ troubled vertices in $V(\rc)$.
Then, we proceed exactly as in the base algorithm:
    we restrict $\rc$ to the troubled vertices
    and apply \RESAMPLE. 

\subsection{Ideas behind LCA realization}

In the base case, the conversion of the global algorithm to LCA is straightforward. 
In fact, the LCA version determines the same area to recolor
    (assuming that both versions process the vertices in the same order).
For the improved algorithm described in the previous subsection,
    conversion to LCA is more complex and alters the behavior of the algorithm.
The main difficulty is that for a bad-component alone that is not initially active,
    it is not easy to quickly decide whether it is going to be activated or not.
There might exist a long chain of activation
        leading to an activation of the considered bad-component,
        and we do not know in which direction to search for the sources of this eventual activation.
Moreover, even if we find out that it will be activated, 
    it is not obvious what the shape of the final component containing it will be,
        since it requires performing activation propagation and 
            determining activation statuses of neighboring bad-components as well.
To address these problems,
    when a troubled vertex of some bad-component is queried,
    we focus on finding an area containing that vertex
        that can be recolored independently from the remaining part of the input hypergraph.
It means that from the beginning of the procedure the component of that vertex is treated as active
    and we allow trimming unsafe edges to that component.
Moreover, we use additional techniques described below
    to limit the expansion of the processed area in a single query.

\subsubsection{Trimming to bad-component}
We extend edge trimming to the case when an unsafe edge $f$
    has at least $\alpha{k}$ troubled vertices in some bad-component $S$, and 
    the set of those vertices together with the accepted vertices of $f$ is not monochromatic.
In such a case, $f$ can be trimmed by removing from it the troubled vertices that do not belong to $S$.
Note that we do not check here whether $S$ is active or not.
The idea behind this step is that from now on $S$ is responsible for the proper coloring of $f$.
If at some point, the colors of the vertices of $S$ get accepted without any resamplings, 
    then $f$ will be obviously colored properly.
Otherwise, if $S$ becomes active, then $f$ will be trimmed anyway,
    and $S$ has enough troubled vertices of $f$ to not break two-colorability of $S$.

\subsubsection{Activation exclusion}
The necessary condition for an inactive bad-component $S$ to be activated is that
    there is an unsafe edge $f$ 
    whose accepted vertices and troubled vertices in $f\cap V(S)$ are of the same color.
When there is no such edge
        or all such edges were trimmed to other components,
    then $S$ cannot be activated.
Therefore if it is not initially active, it stays inactive.
In such a case, we can accept all the proposed colors for the vertices of $S$.
As a result, some unsafe edges become properly colored, and we can treat them as safe.
This, in turn, may enable proving that neighboring bad-components will also not be activated.
The same reasoning can be applied to  a set $C$ of bad-components.
If none of the bad-components in $C$ is initially active
    and there are no unsafe edges 
        intersecting some bad-component outside $C$
        that may activate bad-component from $C$,
    then we can conclude that all bad-components in $C$ remain inactive.

\subsubsection{Conditional expansion}

The idea described in the previous subsection
    can be used for a bad-component to perform some kind
    of search for a potential reason of activation.
If $S_1$ is not initially active, we inspect
    unsafe edges that may cause the activation of $S_1$.
We can select any such $f$, and ask whether other bad-component $S_2$
    intersected by $f$ may become active.
We can continue that procedure as long as there is a risk of activating
    any $S_i$ from the group of bad-components visited so far.
In the end, 
    we either find some initially active component
    or we prove that all the considered bad-components cannot be activated.
It turns out that, if we do not follow the edges that can be trimmed
    with the trimming to bad-component technique,
        then the processed area during such a search is unlikely to be large.

The possibility of finding an initially active bad-component
    can be used in expansion of the component to extend it by a neighboring area.
For a selected bad-component adjacent to the currently constructed eventually final component,
    we launch a search and
        either we find some monochromatic edge (initially active component)
            and extend the component with the whole searched area,
        or convince ourselves that this area cannot be activated.
In the latter case we can simply accept the proposed colors in that area.
In the former we can perform the expansion because 
    the occurrence of a monochromatic edge, as an unlikely event, 
        in a~sense amortizes the expansion of the component.
In fact, we can stop the search procedure not only when we find a monochromatic edge
    but also in a less restrictive case
        when we find an unsafe edge intersecting at least two disjoint bad edges
            outside the search area.
This possibility follows from the technical details of the analysis.

\subsection{LCA procedure} \label{sec:main-lca}

We describe the improved LCA procedure in reference to the base algorithm presented
    in Section \ref{sec:base-lca}.
As previously,
    the ordering of the vertices is constructed dynamically and is driven by the queries and the work of the algorithm.
For a set of edges $S$, by $V_t(S)$ we mean all troubled vertices in $V(S)$.
For an edge $f$, we denote by $f|_t$ the set of troubled vertices of $f$, and by $f|_a$ the set of accepted vertices of $f$.

\subsubsection{\fQuery}

The main procedure is almost identical to its counterpart in the base algorithm
    (Listing \ref{alg:base-lca-query}).
The only difference is that when processing a vertex $v$ of a bad edge,
    it is not only marked as troubled, but also a random color is assigned to $v$.

\subsubsection{\fBuildComponent} \label{sec:main-build-component}

This procedure is the heart of the algorithm
    and is substantially more complex than its analogue in the base version.
It is presented in Listings 
    \ref{alg:main-lca-build-component} and \ref{alg:main-lca-expand-or-accept}
        available in Appendix \ref{sec:listings_improvedLCA}.
It also makes use of subprocedures defined earlier
    (see Listing \ref{alg:base-lca-subprocedures}), with one modification in \FDetermineEdgeStatus
        ~-- once a vertex $w$ is marked as troubled, a random color is also assigned to $w$.
As previously, the procedure works on the line graph of $H$
    and grows a set $B$ of bad edges 
    that will be converted to a final component at the end of the procedure.
It always starts from the bad-component containing the queried vertex $v$,
    and expands it by neighbor bad-components via unsafe edges.
The main change is that
    in the base algorithm each unsafe edge causes expansion of the component,
    here unsafe edges are processed more carefully.
Throughout the procedure we make sure that the size of $B$ does not exceed
    $\vCompBound$ bound on number of edges
    -- if that happens, the procedure stops and declares a failure.

Let $U$ be the set of not processed unsafe edges intersecting~$V(B)$.
If some edge can be trimmed to $V(B)$, it can be safely removed from $U$.
Thus, we may assume that 
    each $f$ in $U$ has fewer than $\alpha{k}$ troubled vertices in $V(B)$.
Since every unsafe edge has more than $\alpha{k}$ troubled vertices,
    each $f$ from $U$ has to intersect at least one bad-component outside $V(B)$.
The procedure applies the following \emph{extension rules} as long as possible:
\begin{itemize}
    \item (r1)
        if there exists $f$ in $U$ that intersects
            at least two disjoint bad edges outside $B$,
    or
    \item (r2)
        if there exists $f$ in $U$
            for which all the vertices of $f$ outside of $V_t(B)$ are monochromatic,
\end{itemize}
then $B$ is 
    extended with all bad edges from the bad-components intersected by $f$;
\begin{itemize}
    \item (r3)
        if there are no edges in $U$ that meet the conditions (r1) or (r2),
        but there exists $f$ in $U$ that
            has fewer than $\alpha{k}$ troubled vertices outside $V(B)$,

\end{itemize}
then call \FExpandOrAccept procedure (described in the following subsection) for $f$,
    which implements the conditional expansion technique,
    and extend $B$ with the returned set of bad edges (which may happen to be empty).

Note that, when there are no edges that meet conditions (r1) or (r2),
    then for any remaining $f$ from $U$
        it is guaranteed that 
            $f$ intersects exactly one bad-component outside $V(B)$
            and $f \setminus V_t(B)$ is not monochromatic.
If such $f$ does not satisfy condition (r3),
    it has at least $\alpha k$ troubled vertices in that external bad-component,
        so it can be trimmed to it (according to trimming to bad-component technique).
Thus, $f$ can be removed from $U$.

After each extension rule, the processed edge is removed from $U$.
On the other hand, when $B$ is extended, new unsafe edges may be added to $U$,
    but we remove those that can now be trimmed to $V(B)$.
Since edges which do not fulfill any of the extension rules are also removed from $U$, 
    finally $U$ becomes empty and the procedure stops.
At this point, $B$ is a set of bad edges
    which are surrounded only by safe and trimmed unsafe edges.

\subsubsection{\fExpandOrAccept}

This procedure is an implementation of the conditional expansion technique,
    through a given unsafe edge $e$.
Similarly to \FBuildComponent, it grows a set $A$ of bad edges,
    which we call a \emph{search area},
        and makes sure that its size does not exceed $\vCompBound$ bound
        (if that happens, the whole algorithm stops and declares a failure).
Initially, $A$ is empty.
Then it becomes expanded by bad-components which may lead to initially active bad-component,
    starting from the not explored bad-component intersected by $e$.
The expansion naturally stops when there are no more candidate bad-components.
The procedure, however, can also stop earlier 
    in case when some monochromatic edge
    or unsafe edge intersecting two disjoint not explored bad edges is found.

Let $Q$ be the set of unsafe edges to be processed (initially it is empty).
Let $C$ be the set of bad edges of the currently expanded bad-component.
Let $U_C$ denote the set of unsafe edges intersecting $V(C)$
    but not adjacent to the edges of $B$ and $A$
    (these are simply those unsafe edges adjacent to the edges in $C$
        that were not explored before expansion of $C$).
The procedure extends $A$ with all edges from $C$,
    and then looks for the following \emph{amortizing configuration}:
    \begin{itemize}
        \item (e1) 
              if $C$ contains monochromatic edge $f$
    \end{itemize}
then 
    the procedure stops and returns set $A$;
    \begin{itemize}
        \item (e2) 
              if $U_C$ contains a monochromatic edge $f$,
        or
        \item (e3) 
              if $U_C$ contains an edge $f$,
              which intersects at least two disjoint bad edges outside $C$,
    \end{itemize}
then
    first set $A$ is extended with all the bad edges of the bad-components intersected by $f$,
    and then the procedure stops and returns $A$.

When no such configuration is found,
    all unsafe edges in $U_C$ are not monochromatic
        and, moreover, each intersects at most one bad-component outside $A$.
We focus on the edges from $U_C$ that can cause an activation of $C$
    -- these are the edges whose 
        troubled vertices in $V(C)$ together with accepted vertices are monochromatic.
Each such an edge $f$ has to intersect exactly one external bad-component
    and troubled vertices of that component together with $f|_a$ ensure a proper coloring of $f$.
If
    there are at least $\alpha k$ troubled vertices of $f$ in that external bad-component,
        $f$ can be trimmed to it (according to the technique of trimming to bad-component).
That is why we add to $Q$
    only those edges from $U_C$ that may cause activation of $C$
        and have fewer than $\alpha k$ troubled vertices outside of $V(C)$.

When processing of $C$ is finished,
    we pick any edge from $Q$ (the set of unsafe edges to be processed) and 
    repeat the above steps for the external bad-component intersected by the selected edge.
It may happen that this component has already been added to $A$,
    in a such case the procedure continues picking edges from $Q$.
When the procedure finishes without encountering amortizing configuration,
    there are no monochromatic edges in $A$
        and all unsafe edges intersecting $V(A)$ are 
            either properly colored by the colors of the accepted vertices and the vertices from $V_t(A)$,
            or are trimmed to bad-components outside it.
Thus, an activation of whole $A$ is excluded.
Then  we mark all vertices in $V_t(A)$ as accepted
        and treat edges properly colored by their colors as safe.
In that case, the procedure returns the empty set.

Note that during this procedure, we do not apply edge trimming to $V(A)$
    when it covers at least $\alpha k$ troubled vertices of some unsafe edge,
        since it can result in a false activation (in case the edge is monochromatic inside $V(A)$).
We also ignore all unsafe edges intersecting $V(B)$
    (they were explored before call to \FExpandOrAccept)
    since, due to not satisfying (r1) and (r2)
        they cannot be used in an amortizing configuration
            or cause an activation
                (it is guaranteed that they are not monochromatic outside $V_t(B)$).

\subsubsection{\fColorComponent}

The last procedure 
    is almost identical to its counterpart in the base algorithm
    (Listing \ref{alg:base-lca-color-component}).
Recall that the only change is at the beginning of the procedure.
Instead of extending $\rc$ with all unsafe edges intersecting it,
    only those unsafe edges that have at least $\alpha{k}$ troubled vertices in $V(\rc)$
        are added.
Then we proceed as in the base algorithm.

\bibliography{references.bib}

\appendix

\newpage
\section{Listings of the improved procedure}  \label{sec:listings_improvedLCA}

\subsection{Listing of the global algorithm}  \label{app:MainGlobal}

\begin{algorithm}[H]
    \newcommand{\AC}{\mathcal{A}}
    \caption{Improved algorithm for uniform hypergraph coloring} \label{alg:main-global}
    \DontPrintSemicolon
    \SetKwProg{Fn}{Procedure}{:}{}
    \Fn{\FHColoring{$H$ - hypergraph}}{
        \codeCommentLine{I. SHATTERING PHASE}
        let $(v_1, v_2, ... v_n)$ be an ordering of $V(H)$ \;
        \For{$i=1$ \KwTo $n$}{
            assign a random color to $v_i$ \;
            \lIf{all edges containing $v_i$ are not bad}{
                mark $v_i$ as \state{accepted}
            } \lElse {
                mark $v_i$ as \state{troubled}
            }
        }
        
        determine status of each $e \in E(H)$  \codeCommentRight{$e$ is bad, safe, or unsafe}
        explore the line graph and build component-hypergraph $H_C = (V_C, E_C)$ \;
        
        \codeCommentLine{activation of bad-components}
        \codeCommentLine{- let $U_C$ be the set of unsafe edges corresponding to $E_C$}
        \codeCommentLine{- let $U(B)$ denote unsafe edges intersecting component $B$}
        \codeCommentLine{- let $U_C(B) = U(B) \cap U_C$}
        \codeCommentLine{- let $V_t(\rc)$ denote set of troubled vertices in component $\rc$}
        $\AC \gets \emptyset$ \codeCommentRight{initialize set of active components}
        $Q \gets \emptyset$   \codeCommentRight{unsafe edges to process}
        
        \codeCommentLine{- initial activation}
        \ForEach{$B \in V_C$}{
            \If{some $e \in E(B)$ or $f \in U(B)$ is monochromatic}{
                mark $B$ as \emph{active} \;
                add $B$ to $\AC$ \aand{} add all edges from $U_C(B)$ to $Q$
            } \lElse {
                mark $B$ as \emph{inactive}
            }
        }
        \ForEach{$f \in U_C$}{
            \lIf{$f$ is monochromatic}{
                merge in $\AC$ all $\rc \in \AC$ intersected by $f$
            }
        }

        \codeCommentLine{- activation propagation}
        \While{
            $Q$ is not empty
        }{
            $f \gets$ next edge from $Q$ (remove it from $Q$)\;
            \If {
                $\forall_{\rc \in \AC} \; |f \cap V_t(\rc)| < \alpha k$
                \aand{}
                $f \setminus V_t(\bigcup\AC)$ is monochromatic
            }  {
                \codeCommentLine{- activate new bad-components through $f$}
                \ForEach{$B \in V_C$ such that $B$ is \emph{inactive} and $f$ intersects $B$}{
                    mark $B$ as \emph{active} \;
                    add $B$ to $\AC$ \aand{} add all edges from $U_C(B)$ to $Q$
                }
                \codeCommentLine{- merge active components through $f$}
                merge in $\AC$ all $\rc \in \AC$ intersected by $f$ \;
            }
        }
        
        \codeCommentLine{II. FINAL COLORING PHASE - color each final component}
        \ForEach{$\rc \in \AC$}{
            \lForEach{$f \in U(\rc)$ such that $|f \cap V_t(\rc) | \geq \alpha k$}{
                add $f$ to $\rc$
            }
            $\rc' \gets$ restriction of $\rc$ to troubled vertices \;
            \RESAMPLE($\rc'$)
        }
  }
\end{algorithm}

\newpage
\subsection{Listing of \fBuildComponent{} (LCA)}

\begin{algorithm}[H]
    \caption{Improved LCA procedure for the final component construction} \label{alg:main-lca-build-component}
    \DontPrintSemicolon
    \SetKwProg{Fn}{Procedure}{:}{}

    \Fn{\FBuildComponent{$v$ - troubled vertex}}{

        $B \gets \emptyset$
            \codeCommentRight{initialize set of bad edges of the component}
        $U \gets \emptyset$
            \codeCommentRight{initialize set of unsafe edges to process}
        $U_{s} \gets \emptyset$
            \codeCommentRight{unprocessed unsafe edges able to launch search}

        $e \gets$ any bad edge containing $v$ \;
        mark $e$ as explored \aand{} run \FExpandBadComp{$e$, $B$, $U$} \;

        \codeCommentLine{process surrounding unsafe edges according to extension rules}
        \While{ 
            $U \neq \emptyset$ or $U_s \neq \emptyset$
        }{
          \While{$U$ is not empty}{
            $f \gets$ next edge from $U$ (remove it from $U$) \;
            \If{$f$ has $< \alpha k$ troubled vertices in $V(B)$}{
                \If{$f$ satisfies rule (r1) or (r2)}{
                    \FExpandViaUnsafe{$f$, $B$, $U$}
                }
                \ElseIf {$f$ can satisfy rule (r3)}{
                    add $f$ to $U_s$
                    \codeCommentRight{$f \setminus V(B)$ has $< \alpha k$ troubled vertices}
                }
            }
          }
          \If{$U_s$ is not empty}{
            $f \gets$ next edge from $U_s$ (remove it from $U_s$) \;
            \If{$f$ has $< \alpha k$ troubled vertices in $V(B)$}{
                \codeCommentLine{$f$ satisfies rule (r3)}
                $(A, U_A) \gets$ \FExpandOrAccept{$f$, $B$, $U$} \;
                $B = B \cup A$ \aand{} \lIf{$|B|> \vCompBound$}{ \FAIL }
                $U = U \cup U_A$ \; 
            }
          }
        } 

        \codeCommentLine{return hypergraph built on set of bad edges}
        \Return $\rc = (V(B), B)$ 
    }
\end{algorithm}

\newpage
\subsection{Listing of \fExpandOrAccept{} (LCA)}

\begin{algorithm}[H]
    \caption{Conditional expansion via unsafe edge $e$ (exploring a search area)} \label{alg:main-lca-expand-or-accept}
    \DontPrintSemicolon
    \SetKwProg{Fn}{Procedure}{:}{}
    \Fn{\FExpandOrAccept{$e$ - unsafe edge}}{
    
        $A \gets \emptyset$
            \codeCommentRight{initialize set of bad edges of the search area}
        $U_A \gets \emptyset$
            \codeCommentRight{initialize set of unsafe edges around search area}
        $Q \gets \{e\}$
            \codeCommentRight{unprocessed unsafe edges allowing expansion}

        \codeCommentLine{process selected surrounding unsafe edges}
        \While{ 
            $Q$ is not empty
        }{
            $f \gets$ next edge from $Q$ (remove it from $Q$)\;

            \codeCommentLine{expand with the external component to which leads $f$}
            $(C,U_C) \gets (\emptyset, \emptyset)$

            \FExpandViaUnsafe{$f$, $C$, $U_C$}
            
            $A = A \cup C$ \aand{} \lIf{$|A|> \vCompBound$}{ \FAIL }
            $U_A = U_A \cup U_C$ \;
            
            \codeCommentLine{inspect new edges -- look for amortizing configuration}
            \If{(e1) is satisfied (there is a monochromatic edge in $C$)}{
                \Return $(A, U_A)$
            }
            \ElseIf{there is an unsafe edge $f$ in $U_C$ satisfying (e2) or (e3)}{
                \FExpandViaUnsafe{$f$, $A$, $U_A$} \;
                \Return $(A, U_A)$
            }

            \codeCommentLine{select edges that may cause an activation}
            \Else{
                \For{$g$ in $U_C$}{
                    \If { $g|_a \cup (g|_t \cap V(C))$ is monochromatic } {
                      \lIf{ $g \setminus V(C)$ has $< \alpha k$ troubled vertices }{
                            add $g$ to $Q$
            }}}}

        } 

        \codeCommentLine{activation exclusion}
        mark all troubled vertices in $V(A)$ as \state{accepted} \;
        \Return $(\emptyset,\emptyset)$
    }
\end{algorithm}

\FullShort{
    \newpage
    \section{Analysis of the base algorithm} \label{sec:base-analysis}

We prove that the algorithm described in Section \ref{sec:base-lca}
    is a $\left(\Oh{\log^2(n)}, \Oh{n}, \Oh{1/n}\right)$-LCA 
    for the problem of two-coloring of uniform hypergraphs.
Below we briefly discuss correctness and running time.
The space requirement $\Oh{n}$ is trivial.
Then, we focus on the probability of a failure in the shattering phase,
    which is the most important part of the analysis.
We prove that, for any sequence of queries, this probability is at most $1/m$.
Moreover, it can be arbitrarily reduced remaining in the same time bound.
Finally, we show that the same is true for the final coloring phase.
Together, that imply $\Oh{1/n}$ bound on a failure of the algorithm.

\subsection{Correctness} \label{sec:base-analysis-correctness}
We show that
    as long as the procedure does not fail, 
    the current partial coloring defined by the colors of accepted vertices
    can be extended to a proper coloring of $H$.
Observe that after each answered query 
    we have the following property:
    each edge is either properly colored (is safe)
        or has at least $\alpha k$ uncolored or troubled vertices.
Thus, if we remove safe edges from $H$
    and restrict it to uncolored and troubled vertices,
    we obtain a hypergraph with
        maximum edge degree at most $\Delta$
        and edges of size at least $\alpha k$ each.
Since it is assumed that $2 \eul (\Delta +1) < 2^{\alpha k}$,
    we know from LLL that remaining hypergraph is two-colorable
        and any proper coloring of that hypergraph can be used 
        to extend the current partial coloring of $H$.

To complete the proof of the correctness in the context of LCA, we should mention that
    the partial coloring build by the algorithm is only extended,
        i.e. once a vertex is marked as accepted its color is never modified.
In particular, vertices queried multiple times receives consistent answers.

\subsection{Running time}
Recall that we treat $k$ and $\Delta$ as fixed (i.e. $\Oh{1}$ wrt the size of the hypergraph),
    so
        checking the current status of each vertex inside any specific edge, 
        or traversing all edges containing some specific vertex, 
        or determining set $N(f)$ for an edge $f$               
    are all made in constant time.
In particular, running time of \FDetermineEdgeStatus is $\Oh{1}$.

Consider a call $\FQuery(v)$.
In case $v$ is accepted or it does not belong to any bad edge,
    the computation is done in constant time.
Otherwise, the algorithm has to build and color the final component containing $v$.
As we argue below, both these steps are completed in $\Oh{\log^2(n)}$ time .

Procedure \FBuildComponent explores set $B$ and all edges adjacent to $B$,
    so at most $(\Delta+1)|B|$ edges are processed in a single call.
Notice that time spent directly on processing a single edge
    (in that procedure and its subprocedures) is constant.
Therefore, due to the explicit bound on the size of $B$,
    this step requires only $\Oh{\log(n)}$ time.

Procedure \FColorComponent starts from the hypergraph $\rc$ containing at most $\vCompBound$ edges,
    so its size $|\rc|$ is $\Oh{\log(n)}$.
In linear time to $|\rc|$ it obtains $\rc'$ whose size is also of the same order.
Then, there are at most $\Oh{\log(n)}$ trials of the \RESAMPLE{} procedure.
A single \RESAMPLE{} requires $\Oh{|\rc'|}$ time for 
    initialization of the colors and finding all monochromatic edges and then, 
    each resampling step can be performed in a~constant time
        (after resampling $f$ we only need to check statuses of $f$ and adjacent edges).
Since we limit number of resamplings to $t_e = |E(\rc')|/\Delta$ which is $\Oh{\log(n)}$
    each trial of \RESAMPLE{} requires $\Oh{\log(n)}$ time.
In total, this gives $\Oh{\log^2(n)}$ running time for the whole procedure.

\nextSection
\subsection{Probability of a failure in the shattering phase} \label{sec:base-shattering-fail}

We focus on the event that during some query, procedure \FBuildComponent
    constructs a set $B$ of bad edges of size larger than $\vCompBound$.
For the purpose of the analysis, suppose that in such a case, instead of failing,
    the procedure continues until a complete final component is constructed.
Then, we measure probability of 
    discovering a large final component,
        i.e. a component that has more than $\vCompBound$ edges.
We analyze that case following the proof by Molloy and Reed (see Lemma 25.2 in \cite{GraphCol}),
    but we improve the estimates slightly.

Employing a common technique,
    with each final component discovered by the algorithm
        we associate a \emph{witness structure}
            of the size that is closely related to the size of the component.
For every witness structure we define an event in such a way that,
    whenever some final component is constructed,
        the event defined by the associated witness structure is satisfied.
Finally, we show that it is very unlikely that
    any of the events defined by all large enough witness structures holds.
This also bounds the probability of discovering a large final component.

Recall that $L(H)$ is the line graph of $H$.
    A set of edges $T \subseteq E$ is called a \emph{$(1,2)$-tree in $L(H)$}, if
        $T$ is connected in the graph in which elements of $T$ are adjacent
            when their distance in $L(H)$ is $1$ or $2$.
    A set of edges $T \subseteq E$ is called a \emph{$(2,3)$-tree in $L(H)$}, if
        $T$ is connected in the graph in which elements of $T$ are adjacent
            when their distance in $L(H)$ is $2$ or $3$,
        and additionally
            no two elements of $T$ are adjacent in $L(H)$.
Equipped with that notation,
    we define as witness structure any $(2,3)$-tree in $L(H)$.
We assign to each witness structure $\tau$,
    an event $\event{\tau}$,
        that every edge in $\tau$ is bad after the shattering phase.
We say that $\tau$ is satisfied when $\event{\tau}$ holds.

\begin{claim} \label{claim:base-bad-edges-12tree}
    If $\rc$ is a final component, then $E(\rc)$ is a $(1,2)$-tree in $L(H)$.
\end{claim}
\begin{claimproof}
It follows easily from the fact that $E(\rc)$ contains edges
    of bad-components that are connected in the component-hypergraph.
Edges of each bad-component by definition are connected in $L(H)$, and
    when two bad-components are adjacent in the component-hypergraph,
    then there must exist two edges, one in each bad-component, which are in a distance 2 in $L(H)$.
\end{claimproof}

For a final component $\rc$
    let $\tau(\rc)$ be the witness structure associated to $\rc$ defined as
        a~maximal subset of $E(\rc)$ which forms a $(2,3)$-tree.

\begin{claim} \label{claim:base-witness-for-component}
    If $\rc$ is a final component,
        then $\tau(\rc)$ is satisfied and $|\tau(\rc)| \geq |E(\rc)|/(\Delta+1)$.
\end{claim}
\begin{claimproof}
By definition $\tau(\rc) \subseteq E(\rc)$, and $E(\rc)$ contains only bad edges,
    therefore $\tau(\rc)$ is satisfied.
The lower bound on its size follows from 
    Claim \ref{claim:base-bad-edges-12tree} and 
    the fact that all elements of $(1,2)$-tree must be adjacent to its maximal $(2,3)$-tree
    (see Fact 25.5 from \cite{GraphCol}).
\end{claimproof}

Claim \ref{claim:base-witness-for-component} allows to focus analysis on witness structures.
Its immediate implication is that,
    if we want to show that
        it is unlikely to discover a final component having more than $\vCompBound$ edges,
    it suffices to show that 
        it is unlikely to find a satisfied witness structure of size greater than $2\log(m)$.

\begin{claim} \label{claim:base-count-witnesses}
    The number of witness structures in $L(H)$ of size $u$ is smaller than $m(4\Delta^3)^u$.
\end{claim}
\begin{claimproof}
It follows from easy counting
    of possible shapes of the tree of size $u$
    (there are at most $4^u$ of them)
       and the number of its embeddings
            into a graph with $m$ vertices and maximum degree at most $\Delta^3$
            (at most $m(\Delta^3)^{u-1}$ ways of doing it).
See also the proof of Claim 25.6 from \cite{GraphCol}.
\end{claimproof}

\begin{claim} \label{claim:base-event-prob}
    Let $p_B = 2^{1-(1-\alpha)k}$.
    For a witness structure $\tau$ of size $u$ we have
    \[
        \pr{\event{\tau}} \leq p_B^{u}.
    \]
\end{claim}
\begin{claimproof}
For a single edge $f$, the probability that $f$ becomes bad is at most $p_B$,
    since during the shattering phase the algorithm has to
        assign at least $(1-\alpha)k$ random colors to vertices of $f$
        (which may not happen)
    and the first $(1-\alpha)k$ of these assignments have to get the same outcome.
For a set of disjoint edges $\tau$, satisfying the event $\event{\tau}$ requires
    such behavior on $u$ disjoint sets of vertices.
\end{claimproof}

\begin{proposition} \label{prop:base-satisfied-witness-prob}
    Let
    \begin{equation} \label{eq:beck-D3-bad-condition}
        \frac{8 \Delta^3}{2^{({1-\alpha}){k}}} < \frac{1}{2}
    \end{equation}
    and $u = \ceil{(c+1)\log(m)}$, for some $c > 0$, then
    \begin{equation*}
       \prt{
         there is a satisfied witness structure of size $u$
       } < \frac{1}{m^c}.
    \end{equation*}
\end{proposition}
\begin{proof}
    We employ the first moment method
        and bound the probability of finding a satisfied witness structure of size $u$,
        by the expected number of such witness structures.
    We use Claims \ref{claim:base-count-witnesses} and \ref{claim:base-event-prob}
        to obtain bound $m(4 \Delta^3 p_B)^u$.
    From assumptions it holds that $4 \Delta^3 p_B < \frac{1}{2}$,
        which together with $u \geq (c+1)\log(m)$ finishes the proof.
\end{proof}

Finally, observe that for all $\alpha \leq 1/4$, and all $k$,
    whenever strengthened LLL condition (\ref{eq:parametrized-LLL-condition}) holds,
    then (\ref{eq:beck-D3-bad-condition}) is satisfied as well.
By taking $u=\ceil{2\log(m)}$ in Proposition \ref{prop:base-satisfied-witness-prob},
    we get that the probability of 
        finding a satisfied witness of size exactly $u$ is less than $1/m$.
It also limits the possibility of finding larger witnesses that are satisfied,
    since each $(2,3)$-tree of size greater than $u$ contains a $(2,3)$-tree of size exactly $u$.
As a result, due to Claim \ref{claim:base-witness-for-component},
    it is unlikely to discover a large final component, and
    the probability of a failure (across all queries)
        during the shattering phase is less than $1/m$.
Moreover, by increasing linearly the acceptable size of the final component,
    we can decrease this bound exponentially.

\nextSectionGap
\subsection{Probability of a failure in the final coloring phase}

Consider a single execution of the procedure \FColorComponent
    and the hypergraph $\rc'$ obtained from the input final component $\rc$.
Since it satisfies Local Lemma assumptions,
    the theorem of Moser and Tardos allows to bound  the expected number of steps
        of \RESAMPLE{} by $|E(\rc')|/\Delta$ (see Theorem 7 in \cite{survey-2021}).
Running that procedure for the number of steps that is twice the expected number,
    ensures that the probability of not finding a solution in a single trial is at most $1/2$.
By repeating the procedure $2 \log(m)$ times,
    the probability of a failure is reduced to $m^{-2}$.

Consider now executions of \FColorComponent across all queries.
Since the number of final components cannot be greater than number of edges,
    it is not called more than $m$ times.
Thus, the probability of a failure (across all queries)
    during the final coloring phase is less than $1/m$.
Again, we point out that by increasing linearly the number of trials of RESAMPLE procedure,
    we can decrease this bound exponentially.

\newpage
\section{Proof of the main result} \label{sec:proof}

In this section we prove our main result formulated in Theorem \ref{thm:main}.
In fact, we show the following, slightly stronger statement.

\begin{theorem}[main result - technical statement] \label{thm:main-technical}
    There exists a local computation algorithm that,
    in polylogarithmic time per query,
        with probability $1-O(1/n)$ solves the problem of two-coloring
        for $k$-uniform hypergraphs with maximum edge degree $\Delta$,
            that satisfy
    \begin{itemize}
        \vspace{0.2cm}
        \item[(i)] $2 \eul (\Delta+1) < 2^{\alpha k}$, for some $\alpha \leq 1/3$, and
        \vspace{0.2cm}
        \item[(ii)] $\Delta \leq \frac{1}{192}~2^{k/3}$.
    \end{itemize}
\end{theorem}

The first condition in the above theorem ensures correctness of our algorithm.
The second follows from the analysis and is required to obtain
    that with high probability every final component is of at most logarithmic size.

Observe that for $k \geq 8$ the second condition implies the first with $\alpha = 1/3$.
On the other hand, it is not reasonable to take smaller $k$,
    since then, the first condition implies that hypergraph has no edges ($\Delta < 1$).
So in principle, we can narrow down theorem assumptions only to the second condition.
In relation to Theorem \ref{thm:main} the technical statement
    allows to take $\alpha = 1/3$ and it clarifies what is meant by a large enough $k$
    (note that the strengthened LLL criterion for $\alpha < 1/3$
        implies the second condition when $k$ is sufficiently large).

Fix $\alpha \leq 1/3$ and let $H$ be a $k$-uniform hypergraph,
    satisfying assumptions of Theorem \ref{thm:main-technical} for that $\alpha$.
Similarly to the analysis of the base algorithm (see section \ref{sec:base-analysis}),
    we show that the improved algorithm described in section \ref{sec:main-lca},
    fulfills the conditions of $\left(\log^2(n), \Oh{n}, \Oh{1/n}\right)$-LCA
    for the problem of two-coloring of $k$-uniform hypergraphs.
To prove the correctness we can use literally the same argument as for the base algorithm
    (see section \ref{sec:base-analysis-correctness}).
The space requirements also remain the same.
Therefore, below we discuss only the running time and the probability of a failure.
Since the improved algorithm
    substantially differs from the base version only in the implementation
        of \FBuildComponent, we focus our analysis on that procedure.

\subsection{Running time of the shattering phase}

Consider a call \FBuildComponent.
Its running time depends linearly on the number of edges processed during its work.
These are bad and unsafe edges explored
    through extension rules and 
    in conditional expansions that end up accepting their searched areas.
Since all processed unsafe edges are adjacent to explored bad edges,
    and $\Delta$ is $\Oh{1}$, we can only focus on estimating the number of the latter.
We show that it is $\Oh{\log^2(n)}$, therefore 
    the final component construction fits in $\Oh{\log^2(n)}$ bound on the running time.

Focus on a single execution of \FExpandOrAccept.
Bad edges are explored when a selected bad-component is expanded,
    then they are added to the search area.
Due to the explicit check of the size during component expansions (in \FExpandBadComp),
    and by checking the size of the search area,
    a single conditional expansion can explore at most $\Oh{\log(n)}$ bad edges.

Back to \FBuildComponent, now we can notice that
    before, during and after each extension rule, set $B$ is of size $\Oh{\log(n)}$.
Thus, it remains to bound the number of bad edges explored in
    those conditional expansions that end up accepting their searched areas.
Notice, that number of all conditional expansions is bounded by
    the number of unsafe edges adjacent to $B$, so there are at most $\Oh{\log(n)}$ of them.
This implies that the total number of explored edges is indeed $\Oh{\log^2(n)}$.

\subsection{Probability of a failure in the shattering phase}

Similarly to analysis of the shattering phase in the base algorithm
    (see section \ref{sec:base-shattering-fail}),
    we suppose that instead of failing,
        procedures \FBuildComponent and \FExpandOrAccept
        continues until a complete component (either final component or search area) is constructed.
Therefore, we measure the probability of discovering a component
    that has more than $\vCompBound$ edges.
We bound this probability in the following proposition.

\begin{proposition} \label{prop:bound-component-size}
    Under the assumptions of Theorem \ref{thm:main-technical}, for any sequence of queries,
        the probability that the algorithm
            discovers a final component or search area containing more than $\vCompBound$ edges
                is smaller than $\frac{2}{m}$.
\end{proposition}

Proof of this proposition has a similar construction to the proof
    of the analogous property for the base algorithm,
        however, the respective steps are more complex than before.
First we define witness structures linked to the line graph of $H$,
    capable of tracking expansion of the final components and search areas.
This time witness structures are more sophisticated and
    events defined by them may express several behaviors on edges belonging to them.
Then, we bound the number of possible structures of a certain size and show that it is unlikely
    to find a large enough satisfied witness structure.
Finally, we show how to construct the satisfied witness structure following
    the work of \FBuildComponent and \FExpandOrAccept procedures.
The whole proof 
    is quite technical that is why we defer it to the next section.

\newpage
\section{Proof of Proposition \ref{prop:bound-component-size}} \label{sec:proof-witness}

\subsection{Witness structures} \label{sec:witness}
In this section we define a witness structure that describes events
    connected to the work of the algorithm presented in section \ref{sec:main-lca}.
The algorithm works on a $k$-uniform hypergraph $H$
    that satisfies assumptions of Theorem \ref{thm:main-technical} for some $\alpha \leq 1/3$.

\subsubsection{Basic properties and definitions}

The main part of a witness structure is a tree that is 
    embedded in the line graph $L(H)$ by a homomorphism (i.e. mapping preserving adjacency).
Additionally, it is equipped with some extra information
    -- labels denoting type of each vertex and an orientation of edges.

\begin{definition}
    \emph{Witness structure of size $u$} is a tuple $\tau= (T, f, \lab, \ort)$ in which
    \begin{enumerate}
        \item $T$ is a tree of size $u$ (unlabeled tree on $u$ vertices)
        \item $f: T \rightarrow L(H)$ is a homomorphism of $T$ into the line graph of $H$,
        \item $\lab$ assigns labels from
              $\{\typeM,
                 \typeB,
                 \typeBin, \typeBout,
                 \typeU,
                 \typeE
              \}$
              to the vertices of $T$,
        \item $\ort$ is a function that directs edges of $T$.
    \end{enumerate}
\end{definition}

\nextLines{
\nextLine We use term \emph{nodes} for the vertices of a witness structure,
            to keep a clear distinction between the vertices of $T$ and the vertices of $H$.
\nextLine For the similar reason, the edges of $T$ are called \emph{arcs}.
\nextLine For each node $x$ let $f_{x}$ denote the edge assigned to it by homomorphism $f$
            (it is possible that two nodes maps to the same edge).
\nextLine Nodes labeled with $\typeE$ are \emph{empty}, all other nodes are \emph{active}.
\nextLine Let $E(\tau) \subset E(H)$ be the set of all edges assigned to nodes of $\tau$.
\nextLine Let $A(\tau) \subset E(\tau)$ be the set of all edges assigned to active nodes of $\tau$.
\nextLine Nodes without outgoing edges are called \emph{roots}.
\nextLine For each node $x$ we define its \emph{depth} $d(x)$ as
            the length of the maximum directed path (of arcs) 
                from $x$ to one of the reachable roots.
\nextLine Roots have depth $0$.
\nextLine For a set of edges $S$, let $N^{+}(S)$ denote set of edges \emph{covered} by the edges in $S$,
    i.e. the set of edges that either belong or are adjacent to $S$
    (formally $N^{+}(S) = \{ S \cup \bigcup_{f \in S} N(f) \}$).
}
For a positive $d$ and a witness structure $\tau$:
\begin{itemize}
    \item let $A_{<d}(\tau) \subset A(\tau)$ be the set of edges assigned to
            the active nodes of $\tau$ of depth strictly smaller than $d$;
    \item let $\tCover{\tau} = V(N^{+}(A(\tau)))$
            be the set of vertices belonging to edges covered by $A(\tau)$,
    \item let $\tCover[<d]{\tau} = V(N^{+}(A_{<d}(\tau)))$
            be the set of vertices belonging to edges covered by $A_{<d}(\tau)$.
\end{itemize}

Definition of witness structure captures the basic shape of the  structure that we are going to use.
That is sufficient to obtain a reasonable bound for the number of structures of a given size.

\begin{claim} \label{claim:count-witnesses}
    The number of witness structures of size $u$ is smaller than $m(48\Delta)^u$.
\end{claim}
\begin{claimproof}
    There are at most $4^u$ unlabeled trees on $u$ vertices.
    For every such tree $T$,
        there are at most $m \Delta^{u-1}$ homomorphic mappings of $T$ into $L(H)$,
        since the first node can be mapped to any vertex of $L(H)$,
            and every other node has to be mapped in a way that preserves adjacency.
    Finally,
        there are $6^u$ candidates for labeling function $\lab$,
        and $2^{u-1}$ choices for $\ort$.
\end{claimproof}

\subsubsection{Events described by witness structures}

For each active node $x$, the label assigned to $x$ determines some basic event.
It depends only on the colors assigned to vertices of $f_x$ in the shattering phase.
Recall that for an edge $f$,
    $f|_a$ denotes set of accepted vertices of $f$,
    and $f|_t$ denotes set of troubled vertices of $f$.

\begin{definition}
    Let $x$ be an active node of a witness structure $\tau$.
    The \emph{basic event $\eventBasic{x}$} is defined as follows,
        depending on the label $\lab(x)$ assigned to $x$:
    \begin{list}{}{\leftmargin=5.2em \labelwidth=5.2em \topsep=0em}
    \newcommand{\litem}[1][]{\item[~\labelitemi~ #1 \hfill --]}

    \litem[$\typeM$] $f_x$ is monochromatic,

    \litem[$\typeB$] $f_x$ is bad,

    \litem[$\typeBin$] $f_x$ is unsafe,
            and set $f_x|_a \cup (f_x|_t \cap \tCover[<d(x)]{\tau})$
            is monochromatic and of size at least $(1-\alpha)k$,

    \litem[$\typeBout$] $f_x$ is unsafe, 
            and set $f_x|_a \cup (f_x|_t \setminus \tCover[<d(x)]{\tau})$
            is monochromatic and of size at least $(1-\alpha)k$,

    \litem[$\typeU$] $f_x$ is unsafe
              with less than $2\alpha{k}$ troubled vertices,
                  that is, $f_x|_a$ is monochromatic and of size at least $(1-2\alpha)k$.
    \end{list}
\end{definition}

Observe that the definitions of basic events for the nodes labeled
    $\typeBin$ or $\typeBout$ involve the shape of the enclosing witness structure.

\begin{definition}
    For a witness structure $\tau$, event $\event{\tau}$ is defined as
        the conjunction of the basic events associated to active nodes of $\tau$.
    We say that witness structure $\tau$ is \emph{satisfied} when $\event{\tau}$ holds.
\end{definition}
Note that empty nodes of $\tau$ do not influence the event $\event{\tau}$
    (alternatively, we may suppose that events assigned to empty nodes are always satisfied).
In the following, we define upper bounds $p_M$, $p_B$, $p_U$ for the probabilities of basic events.

\begin{claim}
    For $x$ labeled with $\typeM$ the probability of $\eventBasic{x}$ is $p_M = 2/2^{k}.$
\end{claim}

\begin{claim}
    For $x$ labeled $\typeB$, $\typeBin$ or $\typeBout$
        the probability of $\eventBasic{x}$ is at most $p_B = 2/2^{(1-\alpha)k}$.
\end{claim}
\begin{claimproof}
    In each case all the accepted vertices of $f_x$ has to be monochromatic
        and when $x$ is labeled $\typeBin$ or $\typeBout$,
            then also some of the troubled vertices have to obtain the same color.
    Observe that for each vertex $v \in f_x$ when a color is sampled for $v$,
        it is known whether it will be marked as accepted or troubled.
    It is also known whether $v$ belongs to $\tCover[<d(x)]{\tau}$ or not,
        since that set is determined precisely by $\tau$.
    Therefore, when a color is chosen for $v$ 
        it is known whether it belongs to that part of $f_x$ that should be monochromatic.
    In order to satisfy $\eventBasic{x}$ there should be at least $(1-\alpha)k$
        of such vertices (which may not happen),
        and since they have to get the same outcome,
        the probability of that event is at most $p_B$.
\end{claimproof}

\begin{claim}
    For $x$ labeled $\typeU$
        the probability of $\eventBasic{x}$ is at most $p_U = 2/2^{(1-2\alpha)k}$.
\end{claim}
\begin{claimproof}
    It follows from the same argument as for the event that $f_x$ is bad,
        except that here the set of accepted vertices can be smaller,
            but there should be at least $(1-2\alpha)k$ of them.
\end{claimproof}

Basic events for the active nodes of some $\tau$ are not independent
    even if the assigned edges are pairwise disjoint.
However, in such a case, 
    it is true that the probability of the conjunction of the basic events 
        is bounded by the product of the corresponding upper bounds.
This is due to the fact that then,
    each basic event specifies a desired behavior on a distinct set of vertices,
    and the upper bounds only concern the probability that
        a specific sampled colors get the same value.

\subsubsection{Proper witness structures}

The nature of the events of our interest allows to impose
    additional properties on the structure.
These properties are used in the derivation of a bound
    on the probability that such witness structure is satisfied.
First,
    the number of empty nodes cannot be to large.
The desired proportion is captured by the following definition.

\newcommand{\numOf}[1]{\#_{#1}}
\newcommand{\numOfSet}[1]{\numOf{\{#1\}}}
\newcommand{\numOfT}[1]{\numOfSet{\type{#1}}}
\begin{definition}
    For a witness structure $\tau$, we define \emph{the balance} of $\tau$ by the following formula
    \[
            2\numOfT{M}
            + \numOfSet{\typeB, \typeBin, \typeBout}
            - \numOfT{E},
    \]
    where $\numOf{X}$ is the number of nodes in $\tau$ with label in set $X$.
    Witness structure with a nonnegative balance is called \emph{balanced}.
\end{definition}
Second,
    we require that each basic event refers to disjoint set of vertices.
These two conditions characterize witness structures of our interest.
\begin{definition}
    A witness structure is \emph{proper} if 
        it is balanced and
        edges assigned to active nodes are pairwise disjoint 
            (in particular, an edge cannot be assigned to two different active nodes).
\end{definition}

In a proper witness structure $\tau$ 
    the probability of $\event{\tau}$ is at most the product of
        our upper bounds on the probabilities of
            the basic events corresponding to the active nodes of $\tau$.
The fact that witness structure is balanced can be used
    to amortize the occurrences of empty nodes in $\tau$ by the presence of active nodes.

\begin{claim} \label{claim:prob-common-bound}
    Let $q = 2^{1-k/3}$.
    For $\alpha \le 1/3$
        it holds that $\max(p_M^{1/3}, p_B^{1/2}, p_U) \le q$.
\end{claim}

The above claim implies that, in some sense,
    each node labeled $\typeM$
        can amortize two empty nodes,
    and each of the nodes labeled $\typeB$, $\typeBin$, or $\typeBout$
        can amortize one empty node.
Using that observation, for a proper witness structure $\tau$,
    we are able to derive an upper bound on the probability of satisfying $\tau$,
        which involves only the size of this witness structure.

\begin{claim} \label{claim:event-prob}
    Let $q = 2^{1-k/3}$.
    For a proper witness structure $\tau$ of size $u$ we have
    \[
        \pr{\event{\tau}} \leq q^u.
    \]
\end{claim}
\begin{claimproof}
    Applying the product of upper bounds of the basic events,
        using Claim \ref{claim:prob-common-bound} and the fact that $\tau$ is balanced,
    we get
    \[
        \pr{\event{\tau}} \leq p_M^{\numOfT{M}}
                               p_B^{\numOfSet{\typeB, \typeBin, \typeBout}}
                               p_U^{\numOfT{U}}
                          \leq q^{3\numOfT{M} 
                                 +2\numOfSet{\typeB, \typeBin, \typeBout} 
                                 + \numOfT{U}}
                          \leq q^u.
    \]
\end{claimproof}

Now we are ready to bound the probability that
    some sufficiently large proper witness structure
        becomes satisfied during the shattering phase.

\begin{proposition} \label{prop:proper-witness-prob}
    Let $q = 2^{1-k/3}$.
    \[
        \prt{ there exists a satisfied proper 
              witness structure of size at least $2\log(m)$ } < \frac{2}{m}
    \]
\end{proposition}
\begin{proof}
    By Claims \ref{claim:count-witnesses} and \ref{claim:event-prob} we get:
    \[
        \Ett{ number of satisfied proper witness structures }{ of size at least $2\log(m)$ }
        < \sum_{u \geq 2\log(m)} m(48 \Delta)^u \cdot q^u.
    \]
    Recall that $\Delta \le \frac{1}{192}~2^{k/3}$ 
        (by assumptions of Theorem \ref{thm:main-technical}).
    Hence $48 \Delta q \le \frac{1}{2}$,
        and the above sum is easily upper bounded by $\frac{2}{m}$.
\end{proof}

\subsection{Construction of satisfied proper witness structures}

We finalize the proof of Proposition \ref{prop:bound-component-size}
    by showing a relation between the work done by the algorithm
        and the occurrence of a satisfied proper witness structure.
Proposition \ref{prop:bound-component-size} is a direct consequence of
    Proposition \ref{prop:proper-witness-prob} and \ref{prop:witness-for-component} (stated below).
Proof of the latter spans to the end of this section.

\begin{proposition} \label{prop:witness-for-component}
    For every component constructed during work of the algorithm
        (either final component or search area),
        there exist a satisfied proper witness structure $\tau$
            of size $|\tau| \geq \frac{|C|}{(\Delta+1)}$,
            where $C$ denote the set of bad edges of that component.
\end{proposition}

In order to prove Proposition \ref{prop:witness-for-component} we show
        that simultaneously with the progress of 
            the algorithm 
        we can construct a satisfied proper witness structure of a related size.
The construction is organized in a series of extensions of the witness structure constructed so far.
That extensions correspond to expansions of the component
    to neighboring bad-components (through unsafe edges).
First, we specify the properties of the witness structure that we intend to preserve during construction.
Then, we describe the basic techniques for constructing and extending a witness structure.
After that, we discuss the possibilities of constructing satisfied proper witness structures
    for the search area explored in \FExpandOrAccept.
Finally, we show how to construct a satisfied proper witness structure for the final component
    obtained in \FBuildComponent.

\subsubsection{Properties of the construction}

\nextLines{
\nextLine For a witness structure $\tau$, the distance between an edge $g$ and $\tau$,
            is the length of the shortest path in $L(H)$ between $g$ and any element of $E(\tau)$.
\nextLine We say that an edge is \emph{adjacent} to $\tau$ if its distance to $\tau$ is exactly $1$.
\nextLine An edge \emph{intersects} $\tau$ if it intersects some edge from $E(\tau)$,
            that is, it belongs to $E(\tau)$ or is adjacent to it. 
\nextLine By \emph{active edges} we mean edges belonging to $A(\tau)$.
\nextLine For a set of edges $S$, let $S|_b \subset S$ denote the set of all bad edges in $S$.
\nextLine We say that $\tau$ \emph{covers} $S$ when $S$ is covered by $A(\tau)$.
}

\newcounter{tempListItemCounter}     
\newcommand{\refa}[1]{(\ref{#1})}    
\begin{definition}[Proper construction] \label{def:proper-construction}
Let $\tau$ be a witness structure and $C$ be a set of edges. 
We say that $\tau$ is \emph{a proper construction for $C$}
    when it fulfills all of the following properties:
    \begin{enumerate}[(a)]
    \item \label{it:disjoint}
        active edges of $\tau$ are disjoint and no edge is assigned to two different active nodes,
    \item \label{it:balanced}
        $\tau$ has a positive balance,
    \item \label{it:satisfied}
        each of the basic events defined by the active nodes of $\tau$ is satisfied,
    \item \label{it:active-covers}
        $\tau$ covers all bad edges contained in $C$,
            that is, $C|_b \subset N^{+}(A(\tau))$,
    \item \label{it:active-inside}
        $A(\tau) \subset C$.
    \setcounter{tempListItemCounter}{\value{enumi}}
    \end{enumerate}
\end{definition}

First two conditions imply that $\tau$ is proper.
The third ensures that it is satisfied.
Condition (\ref{it:active-covers})
    implies that the size of $C|_b$ is at most $(\Delta+1)$ times larger
    than the size of $\tau$ (as required in Proposition \ref{prop:witness-for-component}).
It is also used together with (\ref{it:active-inside}) to prove other properties.
Note that the positive balance of a proper construction ensures that
    the set of active edges is not empty.

In the following section,
    we show how to obtain a proper construction for a set of edges of a single bad-component.
Then, we show how to extend a proper construction $\tau$ for $C$ 
    to become a proper construction for $C \cup D$,
        where $C$ and $D$ are sets of edges that are somehow related.
We call such process as \emph{expanding $\tau$ within $D$}.
We analyze a few typical cases of a relation between $C$ and $D$.
The target construction is achieved by
    a series of additions of new active nodes holding edges from $D$ to $\tau$.
In most cases, a new active node is added together with some empty node.

In order to ensure that the final witness structure is a proper construction for $C \cup D$,
    we stick to the following rules.
In each addition only an edge that belongs to $D$ and 
    is not covered by the current set of active edges may be added as a new active node.
That ensures that conditions \refa{it:disjoint} and \refa{it:active-inside} are met.
The expansion continues until $D|_b$ is not covered,
    which eventually satisfies \refa{it:active-covers}.
After each addition of a group of new nodes we check
    that the balance of a witness structure does not decrease,
    so \refa{it:balanced} is preserved.
Finally, to ensure that condition \refa{it:satisfied} is satisfied,
    for each new active node, we check whether the corresponding basic event holds
    at the moment when that node is added.
With each addition, we also have to ensure that the definitions of the basic events
    assigned to the previously inserted active nodes are not altered.
The events for nodes labeled with $\typeM$, $\typeB$ and $\typeU$ 
    are self-descriptive and cannot change.
A special care must be taken for nodes labeled $\typeBin$ and $\typeBout$.
To avoid altering their definitions, we adhere to the following rule:
\begin{enumerate}[(a)]
\setcounter{enumi}{\value{tempListItemCounter}}
\item \label{it:add-depth}
    when a new active node $x$ is added to $\tau$,
    depth of $x$ in $\tau$ must be greater than the depths of the nodes labeled $\typeBin$ or $\typeBout$
        that correspond to the edges that are in a distance $2$ from $f_x$.
\end{enumerate}
That rule ensures that once a node $x$ with label $\typeBin$ or $\typeBout$
    is added to $\tau$, the set of vertices $f_x \cap \tCover[<d(x)]{\tau}$ does not change
        in the later extensions of $\tau$.

In the following proofs we comply with the above rules.
Therefore, we only explicitly argue that
    the balance is not decreased,
    and that the way of adding a new active node adheres to rule \refa{it:add-depth}.
We also check that
    the basic event assigned to a newly added active node is satisfied,
    however, in case of bad and monochromatic edges 
        it needs no comment that they can be labeled respectively as $\typeB$ and $\typeM$.

\subsubsection{Basic techniques}

\begin{definition}[Adding an edge]
For a witness structure $\tau$ and an edge $g$ that is adjacent 
    to some edge in $E(\tau)$, 
    by \emph{adding $g$ to $\tau$ as $\type{X}$} we mean the following steps.
Let $y$ be a node in $\tau$ such that $f_y$ is adjacent to $g$.
We add a new node $x$, set $f_x=g$, label it $\type{X}$,
    and add to the tree of $\tau$ an arc between $x$ and $y$.
By default, we direct that arc towards $y$.
\end{definition}

\begin{claim}[Proper construction for a bad-component] \label{claim:construct-bad}
    Let $B$ be the set of edges of some bad-component.
    Then, there exist a proper construction for $B$.
\end{claim}
\begin{claimproof}
    We start $\tau$ from a single node labeled $\typeB$ to which we assign an edge $e \in B$.
    Then, as long as $B$ is not covered by $A(\tau)$,
        we pick any $f$ from $B$ that is in a distance $2$ from $A(\tau)$.
    Such an edge has to exist since 
        edges from $B$ that are not covered by $A(\tau)$ are 
            in a distance at least $2$ from that set,
        the distance of $e \in B$ to $A(\tau)$ is $0$,
        and $B$ is connected.
    Edge $f$ is in a distance at most $2$ from $\tau$.
    The following construction ensures that
        edges assigned to empty nodes are adjacent to active edges,
        so $f$ does not belong to $E(\tau)$.
    Thus, we can consider two cases.
    If $f$ is adjacent to $\tau$, we add it to $\tau$ as $\typeB$.
    Otherwise, $f$ is in a distance $2$ from $\tau$.
    In that case, let $g$ be an edge that is on a path of length 2 from $f$ to $\tau$.
    We add $g$ to $\tau$ as $\typeE$,
        then $f$ is adjacent to $\tau$ and we add it to $\tau$ as $\typeB$.

    The initial balance of $\tau$ (after setting the first node for $e$) is $1$.
    The way of extending $\tau$ cannot decrease that balance, 
        so condition \refa{it:balanced} is satisfied.
    All active nodes are labeled $\typeB$
        and only bad edges that belong to $B$ are assigned to them,
        so \refa{it:satisfied} and \refa{it:active-inside} are met.
    When an edge is added as a new active node it is in a distance $2$ from $A(\tau)$,
        so \refa{it:disjoint} holds.
    When the construction eventually stops it also fulfills condition \refa{it:active-covers}.
\end{claimproof}

\begin{definition}[$2$-attainability]
For a witness structure $\tau$ and an edge $g$,
    we say that \emph{$\tau$ is $2$-attainable from $g$},
        when
            $g$ does not intersect $A(\tau)$, and
            for some node $x$ of $\tau$ the distance between $g$ and $f_x$ is $2$.
In such a case,
    by \emph{a deepest node of $\tau$ that is $2$-attainable from $g$},
        we mean a node $x$ of $\tau$ of the greatest depth
            for which $f_x$ is in a distance $2$ from $g$.
In case there are multiple such nodes on the same depth, we choose any of them.
\end{definition}

Note that when some $g$ is not covered by $A(\tau) \neq \emptyset$,
    but its distance to $\tau$ is at most $2$, then $\tau$ is $2$-attainable from $g$.
This follows from the fact that
    some edge in $E(\tau)$ is in a distance at most $2$ from $g$,
    some edge in $A(\tau) \subset E(\tau)$ is in a distance at least $2$ from $g$,
    and $E(\tau)$ is connected.

\begin{definition}[Attaching an edge]
For an edge $g$ and a witness structure $\tau$ that is $2$-attainable from $g$,
    by \emph{attaching $g$ to $\tau$ as $\type{X}$} we mean the following operation.
Let $x \in \tau$ be the deepest node of $\tau$ that is  $2$-attainable from $g$,
    and $h$ be an edge that is on a shortest path from $g$ to $f_x$.
We add $h$ to $\tau$ as empty node
    (even if $h$ is already assigned to some previously added node), joining it towards $x$.
Then, we add $g$ to $\tau$ and label it with $\type{X}$.
\end{definition}
Assuming that $\type{X}$ is neither $\type{E}$ nor $\type{U}$,
    the balance of $\tau$ is not decreased.
This operation adheres to rule \refa{it:add-depth}
    so it does not alter basic events of the active nodes added to $\tau$ before $g$. 
It also retains information about which vertices of $g$ are in $\tCover{\tau}$
    at the moment when $g$ is added to $\tau$.
We express that property in the following claim.

\begin{claim}[Attaching preserves coverage] \label{claim:construct-attach}
    For an edge $g$ and a witness structure $\tau$ that is  $2$-attainable from $g$,
        let $\tau'$ denote the witness structure obtained by attaching $g$ to $\tau$.
    Then,
        $g \cap \tCover{\tau} = g \cap \tCover[<d(z)]{\tau'}$,
    where $z$ denote the newly added node corresponding to $g$ in $\tau'$.
\end{claim}
\begin{claimproof}
    By definition of 2-attainability, the distance of $g$ to $A(\tau)$ is at least $2$.
    Although set $\tCover{\tau}$ is defined by all active edges of $\tau$,
        when we focus on its restriction to $g$,
        only those edges have impact that are in a distance $2$ from $g$.
    By definition of how $g$ is attached to $\tau$,
        all nodes corresponding to such active edges have smaller depth than newly added $z$,
        so these edges belong to $A_{<d(z)}(\tau')$.
\end{claimproof}

In the following we describe how to extend a proper construction $\tau$
    to cover edges of bad-components that are in a close distance to $\tau$.
For the sake of simplicity, in this context we equate a bad-component with its set of edges.

\begin{claim}[Expansion within a bad-component] \label{claim:construct-within-bad}
    Let $\tau$ be a proper construction for a set of edges $C$.
    Let $B$ be a set of edges of a bad-component which is in a distance at most $2$ from $\tau$.
    Then, $\tau$ can be extended to a proper construction for $C \cup B$
        without decreasing the balance.
\end{claim}
\begin{claimproof}

    When $B$ is not covered by $A(\tau)$, it is guaranteed that
        there exist $f \in B$ from which $\tau$ is $2$-attainable.
    As long as this is the case,
        we pick any such $f$ and attach it to $\tau$ as~$\typeB$.
    This operation does not decrease the balance
        and adheres to rule \refa{it:add-depth}.
    The existence of such an edge follows from the fact that $B$ is connected 
        and the initial distance of some $g \in B$ to $\tau$ is at most $2$.
\end{claimproof}

Let $C$ be a set of edges and
    $B$ be a set of edges of some bad-component outside $C$ (that is, $C \cap B = \emptyset$),
    such that there is an unsafe edge $f$ that intersects both $C|_b$ and $B$,
        and $f|_t \subset V(C|_b) \cup V(B)$.
We say that \emph{$f$ may activate $C$ from $B$} when
\begin{itemize}
  \item $f|_a \cup (f|_t \cap V(C|_b))$ is monochromatic and of size at least $(1-\alpha)k$.
\end{itemize}    
We say that \emph{$f$ may activate $B$ from $C$} when
\begin{itemize}
  \item $f|_a \cup (f|_t \setminus V(C|_b))$ is monochromatic and of size at least $(1-\alpha)k$.
\end{itemize}    
These definitions correspond to the way a search area is expanded 
    and to the extension rule (r2) in the construction of a final component.

\begin{claim}[Following a potential cause of activation] \label{claim:construct-bad-inside}
    Let $\tau$ be a proper construction for a set of edges $C$,
        and $B$ a set of edges of a bad-component, such that
        there exist an unsafe edge $f$ which may activate $C$ from $B$.
    Then, $\tau$ can be extended to a proper construction for $C \cup \{f\} \cup B$
        without decreasing the balance.
\end{claim}
\begin{claimproof}
If $B$ is in a distance at most $2$ from $\tau$,
    then we can expand $\tau$ by Claim \ref{claim:construct-within-bad}
        to a proper construction for $C \cup B$.
Note that, since $f$ is an unsafe edge, 
    being a proper construction for $C \cup B$
    is enough to be a proper construction for $C \cup \{f\} \cup B$.
If not, then all the edges of $B$ are in a distance at least $3$ from $\tau$.
In particular, $\tCover{\tau} \cap V(B) = \emptyset$.
On the other hand, $V(C|_b) \subset \tCover{\tau}$ due to condition (\ref{it:active-covers}).
Since $f$ intersects $V(C|_b)$ and $V(B)$,
    it has to be exactly in a distance $2$ from $\tau$.
By assumptions $f$ may activate $C$ from $B$,
    so $f|_a \cup (f|_t \cap V(C|_b))$ is monochromatic.
Since each troubled vertex of $f$ that does not belong to $V(C|_b)$
    has to be in $V(B)$, and the latter set is disjoint with $\tCover{\tau}$,
    it occurs that $f|_t \cap V(C|_b) = f|_t \cap \tCover{\tau}$.
That allows attaching $f$ to $\tau$ as $\typeBin$.
It follows from Claim \ref{claim:construct-attach}
    that the corresponding event is satisfied after the addition of $f$ to $\tau$,
    so we get a proper construction for $C \cup \{f\}$.
Then, we can use Claim \ref{claim:construct-within-bad}
    to expand it within $B$, obtaining a desired proper construction for $C \cup \{f\} \cup B$.
Both edge attachment and expansion according to Claim \ref{claim:construct-within-bad}
    do not decrease the balance of $\tau$.
\end{claimproof}

\begin{claim}[Following a potential activation] \label{claim:construct-bad-outside}
    Let $\tau$ be a proper construction for a set of edges $C$,
        and $B$ a set of edges of a bad-component, such that
        there exist an unsafe edge $f$ which may activate $B$ from $C$.
    Then, $\tau$ can be extended to a proper construction for $C \cup \{f\} \cup B$
        without decreasing the balance.
\end{claim}
\begin{claimproof}
    We proceed analogously to the proof of Claim \ref{claim:construct-bad-inside}.
    The only change is that
        in case $B$ is in a distance at least $3$ from $\tau$,
            we attach $f$ to $\tau$ as $\typeBout$ instead of $\typeBin$.
\end{claimproof}

\begin{claim} [Utilization of an amortizing configuration] \label{claim:construct-use-amort-conf}
    Let $\tau$ be a proper construction for a set of edges $C$.
    Let $f$ be an edge in a distance at most $2$ from $\tau$, such that
        it is either monochromatic
        or intersects at least two disjoint bad edges that are not covered by $A(\tau)$.
    Let $S$ denote the union of sets of edges of bad-components 
        that are not included in $C$, but are intersected by $f$.
    Then, $\tau$ can be extended to a proper construction for $C \cup \{f\} \cup S$
        without decreasing the balance.
\end{claim}
\begin{claimproof}
    When $f$ intersects $\tau$, then $\tau$ can be expanded within each bad-component
        included in $S$ by Claim \ref{claim:construct-within-bad}
        (since these components are then in a distance at most $2$ from $\tau$).
    Hence, we may assume that $f$ is in a distance $2$ from $\tau$,
        so $\tau$ is 2-attainable from $f$.

    In the first case, when $f$ is monochromatic, we attach it to $\tau$ as $\typeM$,
        obtaining a proper construction for $C \cup \{f\}$.
    Then, we expand it within $S$ by Claim \ref{claim:construct-within-bad}.
    In the second case, $f$ intersects at least two disjoint bad edges.
    Let denote them as $g_1$ and $g_2$.
    From assumptions they are not covered by $A(\tau)$ and belong to $S$.
    If any $g_i$ 
        is in a distance at most $2$ from $\tau$, then $\tau$ is 2-attainable from it.
    In such a case we attach $g_i$ to $\tau$ as $\typeB$
        obtaining a proper construction for $C \cup \{g_i\}$.
    In a result, $f$ intersects $A(\tau)$, so all bad-components included in $S$ are now
        in a distance at most $2$ from $\tau$.
    Hence, applying Claim \ref{claim:construct-within-bad} for each of them,
        we expand $\tau$ within $S$.
    If both $g_1$ and $g_2$ are in a distance at least $3$ from $\tau$,
        we attach $f$ to $\tau$ as $\typeE$.
    Then, we add $g_1$ and $g_2$ to $\tau$ as $\typeB$,
        linking them towards the newly added node corresponding to edge $f$.
    In that way the balance is preserved and we adhere to rule \refa{it:add-depth}.
    Afterwards, $\tau$ is a proper construction for $C \cup \{f, g_1, g_2\}$.
    Again, we apply Claim \ref{claim:construct-within-bad} enough times
        to expand $\tau$ within $S$.
    In each case we obtain a proper construction for $C \cup \{f\} \cup S$.
\end{claimproof}

\subsubsection{Construction for a search area}

Consider the search area constructed by the complete execution of $\FExpandOrAccept(e)$
    for some edge $e$
    (recall that for the purpose of analysis, instead of declaring a failure,
        we continue evaluation of the procedure).
We can represent it as $A = C_1 \cup C_2 \cup ... \cup C_t \cup F$, where
\begin{itemize}
    \item $C_i$ is the set of edges of the $i$-th bad-component
            by which the search area was extended,
    \item $t$ is the number of such extensions,
    \item $F$ denotes bad edges added to $A$ 
            when amortizing configuration (e2) or (e3) was found.
\end{itemize}
To be more specific,
    sets $C_i$ correspond to consecutive sets $C$ obtained by \FExpandViaUnsafe,
         skipping calls that returned an empty set.
When no amortizing configuration was found, or (e1) occurs, then set $F$ is empty.
Note that it may also happen in the case of (e2). 
Recall that when $F$ is not empty, it consists of the edges of some bad-components.
Let $U = \{f_1, f_2, ..., f_t\}$ be the set of unsafe edges
    such that $f_i$ triggered expansions of the search area to $C_i$ (in particular $f_1 = e$).
For the case when an amortizing configuration was found,
    let $f_F$ denote the edge that enabled the amortizing configuration.
It intersects $C_t$ and it satisfies either (e1) or (e2) or (e3).
In that case let $C_F = C_t \cup \{f_F\} \cup F$.
When no amortizing configuration was found, $f_F$ is not defined.
Let $\upsilon_F$ denote $\emptyset$ when $F$ is empty,
    and $\upsilon_F = \{f_F\}$ when it is not.
Then, we define $U_A = U \setminus \{f_1\} \cup \upsilon_F$
    as a set of unsafe edges that cause expansions of search area from $C_1$ to $A$.

In the following, we describe how to build a proper construction for a search area
    by tracking the evaluation of \FExpandOrAccept.
Formally, we build a proper construction for $A \cup U_A$.
Then, we show that, in fact, we can obtain a proper construction
    starting from any $C_i$ instead of $C_1$.
Finally we argue, that, in case of finding an amortizing configuration,
    we can obtain a proper construction with a balance at least $2$.

\begin{claim}[Tracking expansions of the search area] \label{claim:construct-search-area-part}
    Let $A_i = C_1 \cup C_2 \cup ... \cup C_i \cup \{f_2, ..., f_i\}$
        be the search area together with unsafe edges after the $i$-th expansion.
    For every $i \leq t$ there exist a proper construction $\tau_i$ for $A_i$.
\end{claim}
\begin{claimproof}
Recall that when $\FExpandOrAccept(e)$ is called, edge $e$
    intersects only one non-explored bad-component,
        and that component is captured by $C_1$.
So for $i=1$, let $\tau_1$ be a proper construction for $A_1 = C_1$ obtained
    according to Claim \ref{claim:construct-bad}.
For the induction step, let $\tau_i$ be a proper construction for $A_i$.
Observe that $f_{i+1}$ potentially may activate $A_{i}$ from $C_{i+1}$.
Thus, we can apply Claim \ref{claim:construct-bad-inside}
    to extend $\tau$ to a proper construction $\tau_{i+1}$ for
    $A_{i+1} = A_{i} \cup \{f_{i+1}\} \cup C_{i+1}$.
\end{claimproof}

\begin{claim}[Proper construction for the search area] \label{claim:construct-search-area}
    There exists a proper construction for $A \cup U_A$.
\end{claim}

\begin{claimproof}
Let $A_i$ be defined as in Claim \ref{claim:construct-search-area-part}.
Then, $A \cup U_A = A_t \cup \upsilon_F \cup F$.
Let $\tau = \tau_{t}$ be a proper construction for $A_t$
    obtained by Claim \ref{claim:construct-search-area-part}.
When $F = \emptyset$ then $\tau$ is already a proper construction for $A$.
Suppose now that $F$ is not empty,
    so the procedure stopped by encountering an amortizing configuration.
In that case $f_F$ is defined at it is an unsafe edge that satisfies either (e2) or (e3)
    (in case of (e1) $F$ is empty).
$F$ consists of the bad edges of the all bad-components intersected by $f_F$
    that are outside $A_t$.
Note that $f_F$ intersects $C_t$ which is covered by $\tau$,
    so its distance to $\tau$ is at most $2$.
Observe that $F$ and $A_t$ are not adjacent,
    and since $A(\tau) \subset A_t$, $F$ is not covered by $A(\tau)$.
That situation fulfills assumptions of Claim \ref{claim:construct-use-amort-conf},
    so we can extend $\tau$ to a proper construction for $A_t \cup \{f_F\} \cup F = A \cup U_A$.
\end{claimproof}

Claim \ref{claim:construct-search-area} proves Proposition \ref{prop:witness-for-component}
    for the case of a search area.
The remaining part of this section is devoted to show some additional properties
    of the construction for a search area.
We focus on the case when the search area is intended to be incorporated into the final component,
    that is, when an amortizing configuration is found.
We show how to make use of such a configuration.

Let us look at the component-hypergraph,
    and focus on vertices corresponding to bad-components
        $V_A = \{ C_1, C_2, ..., C_{t-1}, C_t \}$
    and edges 
        $E_A = \{F_2, ..., F_t\}$
    corresponding to unsafe edges $\{f_2, ..., f_t\}$.
Recall that each $f_i$ intersects only two bad-components of $\{ C_1, C_2, ..., C_{t-1}, C_t \}$.
The first one is such a $C_j$ that added $f_{i}$ to $Q$ -- we denote it by $C_{Q(i)}$.
The other is $C_{i}$.
Thus, $T_A=(V_A, E_A)$ is in fact a tree.
Additionally, let us direct each $F_i$ from $C_{Q(i)}$ towards $C_{i}$.

The proper construction for $A$ obtained in the previous claim,
    is constructed in the order determined by the order 
        in which the components $C_{i}$ are visited by the algorithm.
This order can be viewed as traversing of $T_A$,
    in direction that is consistent with the orientation of the edges $F_i$,
    keeping the set of visited vertices connected.
It turns out, however, that we can expand a proper construction
    within search area by starting from some $C_{i}$ and traversing $T_A$ 
    in any order that preserves connectivity of the set of visited bad-components.
In particular we can allow traversing edges in the order opposite to their direction in $T_A$.

\begin{claim}[Expansion within the search area] \label{claim:construct-area-any-start}
    Let $\tau$ be a proper construction for a such set of edges $C$,
        that for each $C_i$ either $C_i \subset C$ or $C_i \cap C = \emptyset$,
        but for at least one $s \leq t$, $C_s \subset C$.
    Additionally, we require that
        either $F \subset C$ or none of the edges of $F$ is covered by $A(\tau)$.
    Then, $\tau$ can be extended to a proper construction for $C \cup A \cup U_A$
        without decreasing the balance.
\end{claim}
\begin{claimproof}
We start from $C_s$ and traverse $T_A$ in any order,
    keeping the set of visited bad-components connected (disregarding the orientation of edges).
Simultaneously with traversing $T_A$ we extend a proper construction $\tau$
    by visited bad-components, in a similar way to Claim~\ref{claim:construct-search-area-part}.

Let $C'$ be the union of the initial set $C$ and sets $C_i$ processed so far.
When we move from $C_i$ through edge $F_l$ to $C_j$, first we check whether $C_j \subset C$.
If so, we do not have to do anything.
Otherwise, $C_j$ is outside $C'$.
If direction from $C_i$ to $C_j$ agrees with the orientation of $F_l$
    (that is, when $i<j$ and $f_l = f_j$),
    then $f_l$ may activate $C'$ from $C_j$ and 
        we apply Claim~\ref{claim:construct-bad-inside} to extend $\tau$.
When we move in the direction opposite to $F_l$
    (that is, when $i>j$ and $f_l = f_i$),
    then $f_l$ may activate $C_j$ from $C'$ and 
        we apply Claim \ref{claim:construct-bad-outside} to extend $\tau$.

When $T_A$ is traversed, $\tau$ is a proper construction for $C \cup A_t$.
If initially $F = \emptyset$ or $F \subset C$, then we are done.
Otherwise, we can expand $\tau$ within $\{f_F\} \cup F$
    in the same way as in Claim~\ref{claim:construct-search-area}
    (that is, by the application of Claim \ref{claim:construct-use-amort-conf})
    obtaining a~proper construction for the whole $C \cup A \cup U_A$.
\end{claimproof}

\begin{claim}[Proper construction for an amortizing configuration] \label{claim:construct-amortizing-comp}
    There exists a proper construction $\tau_F$ for $C_F$
        that has balance at least $2$ and $f_F$ is covered by $A(\tau_F)$.
    Additionally, all active nodes in $\tau_F$ are labeled $\typeB$ or $\typeM$.
\end{claim}
\begin{claimproof}
    When a search area contains an amortizing configuration then $f_F$ is defined.
    Set $F$ may be empty, however, our proof covers also that case.
    Recall that $C_F = C_t \cup \{f_F\} \cup F$.
    Note that $f_F$ intersects both $C_t$ and $F$ (when it is not empty), so $C_F$ is connected.
    
    Consider first the case when $f_F$ satisfies (e1) or (e2), i.e. it is a monochromatic edge.
    We start $\tau$ from a single node labeled $\typeM$, and assign to it edge $f_F$.
    Its initial balance is $2$.
    Then, we can expand $\tau$ to a proper construction for $C_F$
         in the same way as in Claim \ref{claim:construct-bad}.
    The resulting $\tau$ has at least the same balance as the initial balance,
        thus, as required, it is at least $2$.
    Obviously $f_F$ is covered by $A(\tau_F)$.
    
    Now consider case when $f_F$ satisfies (e3).
    It has to intersect at least $2$ disjoint bad edges in $F$.
    Additionally, it intersects $C_t$ so
        there are at least $3$ disjoint bad edges adjacent to $f_F$.
    Denote them as $e_1$, $e_2$, and $e_3$.
    We can start $\tau$ from a single node labeled with $\typeE$, and assign $f_F$ to it.
    Then, we add $e_1$, $e_2$, and $e_3$ to $\tau$, all labeled with $\typeB$.
    After that initial steps, balance of $\tau$ is $2$ and $f_F$ is covered by $A(\tau_F)$.
    Now, as in the first case, we expand $\tau$ to a proper construction for $C_F$
        in the same way as in Claim \ref{claim:construct-bad},
        without decreasing the balance.
\end{claimproof}

\begin{corollary} \label{cor:construct-area-balance-2}
    For a search area that contains an amortizing configuration,
        there exists a~proper construction with a balance at least $2$.
\end{corollary}
\begin{proof}
    This is a consequence of Claim \ref{claim:construct-amortizing-comp} 
                               and \ref{claim:construct-area-any-start}.
    The former allows to obtain $\tau$ that is a proper construction for 
        $C_F = C_t \cup F \cup \{f_F\}$
        and has a balance at least $2$.
    Then, according to Claim \ref{claim:construct-area-any-start}, 
        it can be extended to the whole search area without decreasing the initial balance.
\end{proof}

\begin{claim}[Incorporating an amortizing configuration] \label{claim:construct-extend-with-amort-conf}
    Let $\tau$ be a proper construction for a set of edges $C$, such that
        $C_F$ in a distance at most $2$ from $\tau$, and
        none of the edges in $C_t \cup F$ is covered by $A(\tau)$.
    Then, $\tau$ can be extended to a proper construction for $C \cup C_F$
        without decreasing the balance.
\end{claim}
\begin{claimproof}
    When $f_F$ intersects $\tau$, then 
        $C_t$ and all bad-components included in $F$ are in a distance at most $2$ from $\tau$.
    Then, using Claim \ref{claim:construct-within-bad} for each such bad-component,
        $\tau$ can be extended to a proper construction for $C \cup C_F$.
    Assume now that $f_F$ does not intersect $\tau$.
    Let $\tau_F$ be the proper construction for $C_F$
        obtained by Claim \ref{claim:construct-amortizing-comp}.
    Notice that by assumptions, edges from $A(\tau)$ are disjoint with edges from $A(\tau_F)$.
    Since $C_F$ is in a distance at most $2$ from $\tau$,
        and $A(\tau_F)$ covers $C_F$ (including $f_F$),
        it occurs that the distance between $\tau$ and $A(\tau_F)$ is at most~$3$.
    
    Let $x$ be a deepest node in $\tau$ such that
        the corresponding edge $f_x$ is in a distance at least~$1$ and at most $3$ from $A(\tau_F)$.
    Let $p$ be a shortest path from $f_x$ to $A(\tau_F)$ and 
        $z$ be a~node in $\tau_F$ such that $f_z$ is on the end of this path.
    Moving from $f_x$ to $f_z$ along $p$,
        for each inner edge we add it to $\tau$ as (E)
            joining it towards the node corresponding to its predecessor in $p$.
    After that, we reorient all the arcs in $\tau_F$ so that $z$ becomes the only root.
    Reorientation of edges does not alter definitions of the basic events,
        since all active nodes in $\tau_F$ are labeled as (B) or (M).
    Finally, we incorporate modified $\tau_F$ into $\tau$ by adding an arc from $z$
        towards such a node $y$ in $\tau$, that $f_y$ is the predecessor of $f_z$ in $p$.
    That way of joining $\tau_F$ to $\tau$ adheres to rule \refa{it:add-depth},
        so all basic events are satisfied in the obtained $\tau$.
    Since the length of $p$ is at most $3$,
        at most $2$ empty nodes are added to $\tau$.
    Their addition is balanced by the balance of $\tau_F$,
        so the balance of $\tau$ does not decrease.
\end{claimproof}

\subsubsection{Construction for a final component}

This is the final step in the proof of Proposition \ref{prop:witness-for-component},
    completing the proof of Proposition \ref{prop:bound-component-size}.
We show how to build a proper construction for a final component
    obtained by the complete evaluation of $\FBuildComponent(v)$.
Formally, it is a proper construction for $B \cup U_e$ where
    $B$ is the set of bad edges of the final component, and
    $U_e$ is the set of unsafe edges that caused expansions of the final component
        (it contains both edges that triggered extension rules and
            edges that are used to expand search areas added to the final component).

We start from a proper construction $\tau$ for the initial bad-component containing $v$.
We obtain it by applying Claim \ref{claim:construct-bad}.
Then we follow the extension rules, and extend $\tau$ accordingly.
Let $B'$ be a set of bad edges in the eventually final component at the moment of extending it.
Let $U'_e$ be a set of unsafe edges that caused expansions so far.
Assume that $\tau$ is a proper construction for $B' \cup U'_e$.
Let $f$ be an unsafe edge that satisfies one of the extension rules.
Note that it is in a distance at most $2$ from $\tau$, since it intersects $B'$.

\begin{claim}[Following rule (r1)]
    Let $f$ satisfy extension rule (r1) and $S$ be the set of edges
        belonging to external bad-components that are intersected by $f$.
    Then, $\tau$ can be extended to a proper construction for $B' \cup U'_e \cup \{f\} \cup S$.
\end{claim}
\begin{claimproof}
In case of (r1), $f$ intersects at least two disjoint bad edges 
    that do not intersect $B' \cup U'_e$, so they are not covered by $A(\tau)$.
This situation is twinned with finding an amortizing configuration (e3) in $\FExpandOrAccept$.
Hence, we can apply Claim \ref{claim:construct-use-amort-conf}
    to extend $\tau$ within $\{f\} \cup S$.
\end{claimproof}

\begin{claim}[Following rule (r2)]
    Let $f$ satisfy extension rule (r2) and $C$ be the set of bad edges
        of the only one bad-component outside $B'$ that is intersected by $f$.
    Then, $\tau$ can be extended to a proper construction for $B' \cup U'_e \cup \{f\} \cup C$.
\end{claim}
\begin{claimproof}
In case of (r2), observe that $f$ may activate $C$ from $B' \cup U'_e$.
That situation is analogous to the extension of a search area
    described in Claim \ref{claim:construct-area-any-start}
    (when we move from $C_{i}$ to $C_j$ in the direction opposite to $F_i$).
We can apply Claim \ref{claim:construct-bad-outside} to expand $\tau$ within $\{f\} \cup C$.
\end{claimproof}

In case of conditional expansion, let $A$ be the set returned by $\FExpandOrAccept(f)$.
When $A$ is empty then no action is required,
    so we focus on the case when it is not, that is, an amortizing configuration is found.
Then, the procedure of expanding a proper construction is more complex.
As in the previous section, we use the form
        $A = C_1 \cup C_2 \cup ... \cup C_t \cup F$
    to represent the search area, and
        $U_A = \{f_2, ..., f_t\} \cup \upsilon_F$
    to denote the set of unsafe edges that triggered consecutive expansions.
Each $f_i$ leads from $C_{Q(i)}$ to $C_i$, and
    $f_F$ denotes the edge that enabled the amortizing configuration.
We also use $C_F = C_t \cup \{f_F\} \cup F$.

\begin{claim}[Following rule (r3)]
    Let $f$ satisfy extension rule (r3) and $A$ be the search area
        returned by $\FExpandOrAccept(f)$, which contains an amortizing configuration.
    Then, $\tau$ can be extended to a proper construction for
        $B' \cup U'_e \cup \{f\} \cup A \cup U_A$.
\end{claim}
\begin{claimproof}
Depending on the distance between $A$ and $\tau$ we distinguish two cases.
Suppose that $A$ or $U_A$ is in a distance at most $2$ from $\tau$.
Then, we extend $\tau$ within at least one $C_i$ in the following way:
\begin{itemize}
\item if $C_F$ is in a distance at most $2$ from $\tau$
        then we extend $\tau$ to a proper construction for $B' \cup U'_e \cup C_F$
        by using Claim \ref{claim:construct-extend-with-amort-conf};
\item otherwise, if some $C_i$ is in a distance at most $2$ from $\tau$
        then we extend $\tau$ to a proper construction for $B' \cup U'_e \cup C_i$
        by using Claim \ref{claim:construct-within-bad};
\item otherwise, there is some $f_i$ that is in a distance at most $2$ from $\tau$
        and none of the edges in $C_{Q(i)}$ and $C_i$ is covered by $A(\tau)$;
        in that case $f_i$ intersects at least two disjoint bad edges
            (one in $C_{Q(i)}$ and the other in $C_i$)
        so we use Claim \ref{claim:construct-use-amort-conf},
            to expand $\tau$ within $\{f_i\} \cup C_{Q(i)} \cup C_i$.        
\end{itemize}
After this step, $\tau$ fulfills assumptions of Claim \ref{claim:construct-area-any-start},
    so we extend it to cover the whole search area.

Suppose now that all edges in $A$ and $U_A$ are in a distance at least $3$ from $\tau$.
In that case, $f$ has to be in a distance $2$ from $\tau$.
Let $p$ be a path $e_1 \in B', f, e_2 \in A$, from $B'$ to $A$.
Note that, $e_1 \notin E(\tau)$.
Let $\tau_A$ be a proper construction for $A \cup U_A$ with the balance at least $2$,
    obtained by Corollary \ref{cor:construct-area-balance-2}.
Since $e_2$ is covered by $\tau_A$, $f$ is in a distance at most $2$ from $\tau_A$.

If $f$ intersects $A(\tau_A)$, then let $x$ be an active node in $\tau_A$
    such that $f$ is adjacent to $f_x$.
We add $e_1$ to $\tau$ as empty node.
Then $f$ is adjacent to $\tau$ so we add it to $\tau$ as empty node $y$.
Finally, 
    we incorporate $\tau_A$ into $\tau$ by adding an arc from $y$ towards $x$.

In the remaining case, when $f$ is disjoint with all active edges of both constructions,
    note that both $\tau$ and $\tau_A$ are $2$-attainable from $f$.
Observe now that $f$ has less than $\alpha k$ troubled vertices in $B$
    and the same is true for $A$.
Since $f$ does not have any other troubled vertices, it can be labeled $\typeU$.
We create a node $x$ for $f$ and attach it to $\tau$ as $\typeU$.
Then, we incorporate $\tau_A$ into $\tau$, by attaching $f$ (using the same node~$x$) to $\tau_A$.
Note that each attachment inserts one empty node.

In both cases the excess balance in $\tau_A$
    amortizes two extra empty nodes used to connect $\tau_A$ to $\tau$,
        so the positive balance of $\tau$ is preserved.
The depth of active nodes remains unchanged after incorporating $\tau_A$ into $\tau$.
Since $A(\tau_A) \subset A \cup U_A$ which is in a distance at least $3$ from $\tau$,
    no definition of the basic events is altered after incorporation of $\tau_A$.
Also addition of $f$ as $\typeU$ in the last case,
    does not alter definitions of the basic events.
This means that finally,
    $\tau$ is a proper construction for $B' \cup U'_e \cup \{f\} \cup A \cup U_A$.
\end{claimproof}

}{}

\end{document}